\documentclass[12pt]{iopart}
\expandafter\let\csname equation*\endcsname\relax
\expandafter\let\csname endequation*\endcsname\relax
\usepackage{amsmath,amssymb, amscd}
\usepackage{blindtext}
\usepackage{pdfsync}
\usepackage{graphicx,subfigure,wrapfig,sidecap}
\usepackage{color}
\numberwithin{equation}{section}
\def\strutdepth{\dp\strutbox}
\def\nw#1{\strut\vadjust{\kern-\strutdepth\vtop to0pt{\vss\hbox to\hsize
{\hskip\hsize\hskip5pt$\leftarrow$\hss\strut}}}{\em #1}}
\newtheorem{theorem}{Theorem}[section]

\newenvironment{proof}{\noindent {\it Proof }}{\hfill $\square$}

\newcommand{\eqa}{\begin{eqnarray}}
\newcommand{\eeqa}{\end{eqnarray}}
\newcommand{\beq}{\begin{equation}}
\newcommand{\eeq}{\end{equation}}

    \newtheorem{proposition}[theorem]{Proposition}
    \newtheorem{conjecture}{Conjecture}
    \newtheorem{Definition}[theorem]{Definition}
    
    \newtheorem{Remark}[theorem]{Remark}
    \newenvironment{remark}{\begin{Remark}\rm}{\end{Remark}}
    \newtheorem{Example}[theorem]{Example}
    
    \newtheorem{Assumptions}[theorem]{Assumptions}

\begin{document}
\title[Shock formation]{ Shock formation in the 
dispersionless Kadomtsev-Petviashvili equation }

\author{
T. Grava$^{1,2}$, C. Klein$^3$, and J. Eggers$^1$
}

\address{
$^1$School of Mathematics, 
University of Bristol, University Walk,
Bristol BS8 1TW, United Kingdom  \\
$^2$SISSA, Via Bonomea 265, I-34136 Trieste, Italy \\
$^3$Institut de Math\'ematiques de Bourgogne, Universit\'e de Bourgogne, 
9 avenue Alain Savary, 21078 Dijon Cedex, France
    }

\begin{abstract}
The dispersionless Kadomtsev-Petviashvili (dKP) equation 
$(u_t+uu_x)_x=u_{yy}$ is one of the simplest nonlinear wave equations 
describing two-dimensional shocks. To solve the dKP equation we use a 
coordinate transformation 
inspired by the method of characteristics for the one-dimensional Hopf 
equation $u_t+uu_x=0$. We show numerically that the 
solutions to the transformed equation  develops singularities at later times with respect to the solution of the dKP equation.
This permits us to extend the dKP solution as the graph of a 
multivalued function beyond the critical time  when the  gradients blow up.
This overturned solution is multivalued in a lip shape region in the $(x,y)$ 
plane, where the solution of the dKP equation exists in a weak sense only, 
and a shock front develops. A local  expansion reveals the universal 
scaling structure of the shock, which after a suitable change of coordinates 
corresponds to a generic cusp catastrophe.
We provide a heuristic derivation of the shock front position near the 
critical point for the solution of the dKP equation, and study the 
solution of the dKP equation when a small amount of dissipation is added. 
Using multiple-scale analysis, we show that in the limit of small 
dissipation and near the critical point of the dKP solution, the 
solution of the dissipative dKP equation converges to a Pearcey integral.
We test and illustrate our results by detailed comparisons with 
numerical simulations of both the regularized equation,  the dKP 
equation, and the asymptotic description given in terms of the Pearcey 
integral. 
\end{abstract}

\maketitle

\section{Introduction}
\label{sec:intro}
Perhaps the best known example of a singularity in an 
evolution equation is the formation of jump discontinuities of
the density and of the velocity field in the Euler equations of
compressible gas dynamics. As these discontinuities propagate, 
they are known as shock waves. In the case of a planar shock front,
the problem can be reduced to a one-dimensional equation for the 
velocity alone \cite{LL84a} (a so-called simple wave). The resulting
wave profile overturns to form an s-shaped curve, the 
point where the gradient first becomes infinite (known as the 
gradient catastrophe) corresponds to the formation of a shock. 
From the overturned solution the physical solution can be 
reconstructed by inserting a jump discontinuity (the shock). The shock 
solution is a weak solution of the equation, which satisfies 
additional conditions motivated by physical considerations \cite{Lax57}.
This shock solution is also found by taking the limit of vanishing
viscosity in the dissipative form of the equations, yielding a weak 
solution (see \cite{Bressan2005} for conservation laws in one space 
dimension and \cite{K70,Metivier2008} for hyperbolic 
equations in several space dimensions). 

The  existence of such gradient catastrophe  points has been proved in 
\cite{Alinhac2002,Majda84} for hyperbolic equations in many 
space dimensions. However, to the best of our knowledge, if the initial 
condition depends on two or three spatial variables, little is known about
the two or three-dimensional spatial structure of the shock near the 
blow-up points of the gradients. In particular, it would be interesting 
to know the 
self-similar structure of the solution both before and after shock
formation \cite{EF09}. A rare instance of where we have a more 
or less complete understanding of a higher dimensional singularity 
is the spatial structure of caustics of wave fronts in the 
approximation of geometrical optics \cite{Arnold90b,Nye99}. Two-dimensional
wave breaking has also been studied in \cite{PLGG08}, using a simple
kinematic equation, for which an exact implicit solution is 
available. 

In this paper, we study the formation of two-dimensional shocks 
in a simple nonlinear wave equation known variously as 
the dispersionless Kadomtsev-Petviashvili (dKP) equation \cite{KP70}, 
or the Zabolotskaya-Khokhlov (ZK) equation \cite{ZK69}. The equation 
has the advantage that its one-dimensional form, the Hopf equation, 
has only one family of characteristics. The dKP equation can be 
seen as a long wavelength version of the original Kadomtsev-Petviashvili 
(KP) equation \cite{KP70}:
\beq
\label{KP}
(u_t + uu_x + u_{xxx})_x = \pm u_{yy},
\eeq
but with the highest order dispersive term $u_{xxx}$ dropped,
namely
\[
(u_t + uu_x )_x = \pm u_{yy}.
\]
The subscript denotes the derivative with respect to the variable. 
With a $+$ sign on the right hand side, (\ref{KP}) is 
known as the KPI equation, or as the KPII equation in the opposite 
case. However, in the case of the dKP equation the two signs are equivalent 
under the transformation $u\rightarrow -u$ and $x\rightarrow -x$, and for
the remainder of this manuscript we will consider only the positive sign. 
Depending on context, the KP  equation describes wave profiles for 
layers of inviscid fluid of finite depth, waves in plasmas, or 
the propagation of sound beams in nonlinear media. 

While the Cauchy problem for the KP equation is
globally well-posed in a suitable space \cite{Saut2002}, by dropping the 
dispersive term, the dKP equation becomes a nonlocal scalar conservation 
law in two space variables.  Even for smooth initial data, the solution 
remains smooth only for finite time. In \cite{Rozanova} it is shown that 
the solution of the dKP equation is locally well posed in the Sobolev space 
$H^s$, $s>2$, so that for $s\geq 4$ one has classical solutions.
Particular solutions of the dKP equations have been obtained with several 
techniques \cite{Ferapontov,GK89,Kono,DMT, Raimondo}. The 
Cauchy problem for the dKP equation and shock formation have been 
studied recently in \cite{MS06, MS08,MS12, MS11b}, using the inverse scattering 
transformation, which relies on the integrability of the dKP inherited 
from the KP equation \cite{TT1995,ZaManakov}.

To sketch a derivation of the dKP equation, we follow the original derivation 
of the KP equation \cite{KP70}. We start from the Hopf equation 
\begin{equation}
\label{H}
u_t+uu_x = 0
\end{equation}
for a wave field $u$, with only a convective non-linearity. This is 
the simplest model equation describing wave steepening and shock
formation. In a frame of reference moving at the sound speed $c$,
a simple wave can be shown to be described by (\ref{H}) \cite{LL84a}.
Assuming a weak $y$-dependence, we add a small correction 
$\psi$ on the right hand side of (\ref{H}); 
\beq
u_t+uu_x = \psi. 
\label{H_lin}
\eeq
For a wave of small amplitude, the second term in the above equation 
can be neglected. Assuming a dispersion relation 
$\omega = k c = \sqrt{k_x^2+k_y^2}c$, one obtains 
in a frame of reference moving along the $x$-axis with velocity  $c$ that 
$\omega = kc - k_x c \approx c k_y^2/(2 k_x)$. For (\ref{H_lin}) to match
this dispersion relation, we must have $\psi_x \approx c u_{yy}/2$.
Taking the $x$-derivative on both sides of (\ref{H_lin}) we obtain 
\[
(u_t+uu_x)_x = -\frac{c}{2}u_{yy}.
\]
Rescaling $x\to -x$ $u\to -u$ and $y\to \sqrt{2/c}y$, one arrives at 
the equation
\beq
\label{dKP_int}
(u_t+uu_x)_x = u_{yy}.
\eeq
Note that in spite of its name, the dKP equation (\ref{dKP_int}) 
contains dispersion, and only the highest order dispersive term has 
been dropped relative to (\ref{KP}). Other contexts in which 
(\ref{dKP_int}) is used are described in \cite{BMP90}. 

The Hopf equation (\ref{H}) is solved by observing that 
the velocity is constant along characteristic curves $x(\xi,t)$,
given by \cite{Chorin_Marsden}: 
\begin{equation}
\label{char_H}
x(\xi,t) = u_0(\xi) t + \xi. 
\end{equation}
Thus for any initial condition $u(x,0) = u_0(x)$, one finds an exact
solution $u(x,t) = u_0(\xi)$ in implicit form. Wave breaking occurs 
when two characteristics cross, which always occurs when the initial 
condition has negative slope. A shock first forms along the characteristic
originating from the point $\xi_c$ of greatest negative slope by absolute 
value, where the solution $u(x,t)$ has a point of blow up of the gradient.

Thus if one expands the initial condition about $\xi_c$,
one finds that the profile assumes a characteristic 
s-shape \cite{EF09}:
\begin{equation}
\label{S}
\Delta x - \Delta u\Delta t + t_c^4 u'''(\xi_c)\Delta u^3/6 = 0,
\end{equation}
where $\Delta u = u - u_c$, and $\Delta x=x-x_c-u_c(t-t_c)$. 
For $\Delta t = t-t_c > 0$ (after shock formation), 
the profile has become multivalued. Balancing the three terms in 
(\ref{S}), one sees directly that $\Delta u$ must be of order 
$\Delta t^{1/2}$, and so $\Delta x$ of oder $\Delta t^{3/2}$
\cite{PLGG08,EF09}. 

If one solves (\ref{dKP_int}) with an initial condition which depends 
on $y$, the equation can no longer be solved with the method of characteristics. 
The idea underlying this paper is that the dependence on the 
$y$-coordinate is weak, so the structure of the solution is essentially
the same as before, but the different stages of overturning are ``unfolded''
in the $y$-direction \cite{MN14}. This means that effectively the singularity
time  becomes a function of $y$. If we choose the origin such that a 
singularity occurs at $y=0$ first, and expand $t_c$ in a Taylor 
series near $y=0$, we obtain $t_c(y) = t_c(0) + a y^2 + O(y^3)$, with 
$a > 0$ a constant and $t_c(0)\equiv t_c$. This means
that $\Delta t = t - t_c - ay^2\equiv \bar{t} - ay^2$, and 
the two-dimensional wave breaking is governed by the scalings 
\beq
\Delta u \sim \bar{t}^{1/2}, \quad \Delta x \sim \bar{t}^{3/2}, 
\quad \Delta y \sim \bar{t}^{1/2}. 
\label{scaling}
\eeq
In this paper we will show that the scalings (\ref{scaling}) indeed 
describe the similarity structure of wave breaking in the dKP equation. 

The estimates (\ref{scaling}) imply that $\Delta y \gg \Delta x$ near the 
shock, consistent with our assumption of a slow variation in the 
$y$-direction. The central idea of our paper is to use this insight
to generalize the characteristic transformation (\ref{char_H}) to 
allow for a slow $y$-dependence: 
\begin{equation}
\label{implicit0}
\left\{
\begin{array}{rl}
& u(x,y,t)=F(\xi,y,t)\\
& x=tF(\xi,y,t)+\xi
\end{array}\right.
\end{equation}
Applying transformation (\ref{implicit0}) to (\ref{dKP_int}) results in a 
PDE for $F(\xi,y,t)$ which we will study in the next section (see equation (\ref{eqF}));
 the initial condition for $F$ is given by 
\begin{equation}
\label{initial}
F(x,y,0)=u_0(x,y).
\end{equation}
Note that if the initial data  $u_0(x,y)$ has no $y$-dependence, 
(\ref{implicit0}) yields the exact characteristic solution with
$F(x,y,t) = u_0(x)$ as described before; in particular, $F$ is $y$ and  
time-independent. As in the method of characteristics, the solution 
$u(x,y,t)$ of the dKP equation encounters a gradient catastrophe when 
the transformation  $x=tF(\xi,y,t)+\xi$ defining $\xi=\xi(x,y,t)$  is 
not invertible, namely when $tF_\xi(\xi,y,t)+1=0$. Our numerical results
show that as a result of the unfolding (\ref{implicit0}), the function 
$F(\xi,y,t)$ remains regular at the time $t_c$ of shock formation of the 
solution $u(x,y,t)$ of the dKP equation. Moreover, our numerics indicate
the derivatives of $F$ remain bounded for times substantially beyond $t_c$. 
However, since $F$ satisfies a nonlinear equation (see (\ref{eqF}) below), we 
believe that $F$ will typically develop a singularity for some time $t>t_c$;
we give an example of such a singularity in a particular case. 

Manakov and Santini \cite{MS08,MS12} have proposed a transformation 
for analysing
the gradient catastrophe of dKP equation which is superficially similar 
to ours, which is motivated by the inverse scattering transform.  
Their transformation differs from ours by a factor of $2$ in front of 
the unfolding term: 
\begin{equation}
\left\{
\begin{array}{ll}
&u(x,y,t)={\tilde F}(\zeta,y,t)\\
&x=2t{\tilde F}(\zeta,y,t)+\zeta\\
&{\tilde F}(\zeta,y,0)=u_0(\zeta,y)
\label{F_MS}
\end{array}
\right.
\end{equation}
as a result, the 
transformation does not unfold the overturned profile if there is no 
$y$-dependence. In fact, transformation of the Hopf equation leads to 
the same equation ${\tilde F}_t-{\tilde F}\tilde{F}_{\zeta}=0$ as before, 
but with propagation in the opposite direction, and with the same initial 
data  $\tilde{F}(\zeta,0)=u_0(\zeta)$. This means that for $y$-independent
initial data localized in the $x$-direction, $\tilde{F}(\zeta,t)$ 
will experience a gradient catastrophe {\it before} $u(x,t)$ does, if 
the initial profile is steeper on the left than on the right. The same 
remains true for solutions of the full dKP equation with localized 
initial data: we checked numerically that for the initial data 
considered in this manuscript, i.e. the $x$-derivative of a Schwartz 
function, the function $\tilde{F}(\zeta,y,t)$ suffers a gradient catastrophe 
{\it before} a gradient catastrophe occurs in the original profile $u(x,y,t)$. 

To further illustrate the difference between the two parameterizations,
note that combining (\ref{implicit0}) and (\ref{F_MS}) one finds 
$F$ in terms of ${\tilde F}$:
\begin{equation}
\left\{
\begin{array}{ll}
&F(\xi,y,t)={\tilde F}(\zeta,y,t)\\
&\xi=t{\tilde F}(\zeta,y,t)+\zeta,
\label{F_MS1}
\end{array}
\right.
\end{equation}
or ${\tilde F}$ in terms of $F$:
\begin{equation}
\left\{
\begin{array}{ll}
&{\tilde F}(\zeta,y,t)=F(\xi,y,t)\\
&\zeta=-tF(\xi,y,t)+\xi.
\label{F_MS2}
\end{array}
\right.
\end{equation}
If we assume that $\tilde{F}(\zeta,y,t)$  has no singularities and that 
$2t{\tilde F}_\zeta(\zeta,y,t)+1>0$ in some time interval $[0,t']$, then 
it follows from (\ref{F_MS}) that the solution  $u(x,y,t)$ of the dKP 
equation is regular in the same time interval. But since we also have
$t{\tilde F}_\zeta(\zeta,y,t)+1>0$, it follows from (\ref{F_MS1}) that 
$F(\xi,y,t)$ is regular in $[0,t']$ as well. 

On the other hand, assuming that $F(\xi,y,t)$ is regular and 
$tF_\xi(\xi,y,t)+1>0$ in some time interval $[0,t']$, it follows 
from (\ref{implicit0}) that once again $u(x,y,t)$ is regular in $[0,t']$. 
However, this does not imply that ${\tilde F}(\zeta,y,t)$ is regular, 
since it may happen that $-tF_\xi(\xi,y,t)+1=0$ for some $t\in(0,t']$, 
even though $tF_\xi(\xi,y,t)+1>0$ for all $t\in[0,t']$. This argument shows that 
$\tilde{F}(\zeta,y,t)$, 
as defined by (\ref{F_MS}),  might encounter singularities even before 
$u(x,y,t)$ does. 

Our formulation allows us to find spectrally accurate solutions to 
$F=F(\xi,y,t)$, 
from which $u(x,y,t)$ can easily be reconstructed. The alternative would be 
to use numerical methods for hyperbolic equations which remain
stable even after the formation of shocks \cite{Leveque92}. However, 
these methods introduce numerical dissipation near the shock, which 
renders the solution inaccurate. These sources of inaccuracy can be 
avoided using our transformation. 
The  main results of this paper  are the following:
\begin{itemize}
\item in section 2 we describe the solution of the dKP equation by using 
a transformation inspired by the method of characteristics and by \cite{MS08}.
This transformation reduces the Cauchy problem for the dKP  equation to 
the Cauchy problem for the function $F(\xi,y,t)$ introduced in 
(\ref{implicit0}), which is regular beyond $t_c$. 

\item in section 3 we study the singularity formation in the solution 
to the dKP equation as done in \cite{MS08},\cite{MS12}.  We then show that the local structure of  the dKP solution near 
the point of gradient catastrophe, in a suitable system of coordinates, is equivalent to the unfolding of an $A_2$ singularity. 
We derive the self-similar structure of the lip-shaped 
domain where the solution  of the dKP equation becomes multivalued.
\item In section 4 we give a heuristic derivation of the shock front 
position near the critical point of the solution of the dKP equation,
and study the solution of the dKP equation when dissipation is added
(called the dissipative dKP equation). Using multiple-scale analysis, 
we show that in the limit of small dissipation and near the critical 
point of the dKP solution, the solution of the dissipative dKP equation 
converges to a Pearcey integral. 
\item In section 5 we compare our analysis with detailed numerical 
simulations. Solutions for initial data with and without symmetry with 
respect to $y\mapsto -y$ are studied. It is shown that our
numerical approach allows to continue dKP solutions to a second 
gradient catastrophe, well after the first catastrophe has occurred.
We find no indication for blow-up of the solution to the transformed dKP 
equation.  
\end{itemize}

\section{Solution by characteristic transformation}
\label{sec:char}
We  consider the Cauchy problem  for the dKPI equation 
\begin{equation}
\label{dKP}
\left\{
\begin{array}{ll}
(u_t+uu_x)_x&=u_{yy},\\
u(x,y,t=0)&=u_0(x,y),\quad x,y\in\mathbb{R},\;t\in\mathbb{R}^+. 
\end{array}\right.
\end{equation}
Since we are interested mainly in local properties of the 
solution, we will assume that $u_0(x,y)$ is in the Schwartz class, namely it is  smooth and decreases rapidly 
at infinity. Equation (\ref{dKP}) can also be written in the evolutionary form
\beq
\label{dKP1}
u_t+uu_x=\partial_x^{-1}u_{yy},
\eeq
where $\partial_x^{-1}f(x)\equiv\int_{-\infty}^x f(x')dx'$. This has the form 
of a nonlocal conservation law 
\beq
\label{cons1}
u_t + \nabla{\bf f} = 0, \quad 
{\bf f} = \frac{u^2}{2}{\bf e}_x - \partial_x^{-1}u_y{\bf e}_y,
\eeq
with ${\bf e}_x$ and ${\bf e}_y$ unit vectors in the $x$ and $y$ directions.
As a result,
\beq
\label{p}
\int_{\mathbb{R}^2} u(x,y,t)dxdy=\int_{\mathbb{R}^2} u_0(x,y)dxdy.
\eeq
Similarly,
\beq
\label{cons2}
(u^2)_t + \left[2 u^3 / 3 - \left(\partial_x^{-1}u_y\right)^2\right]_x
+ \left(2u\partial_x^{-1}u_y\right)_y = 0,
\eeq
and hence the $L^2$ norm is also a conserved quantity:
\begin{equation}
\label{mass}
M(t)\equiv\int_{\mathbb{R}^2} u^2(x,y,t)dxdy=\int_{\mathbb{R}^2} u_0^2(x,y)dxdy.
\end{equation}

Since the left hand side of (\ref{dKP}) is a total derivative, solutions 
have to satisfy the constraint
\[
\int_{\mathbb{R}}u_{yy}(x,y,t)dx=0,\quad t>0.
\]
If the initial data does not satisfy such constraint, a low decays at infinity occurs for $t>0$ even for initial data in the Schwartz class.
This is a manifestation of the infinite speed of propagation 
in the dKP equation.  For this reason we choose initial data such that 
\begin{equation}
\label{constID}
\int_{\mathbb{R}}u_0(x,y)dx=0 
\end{equation}
for all $y$,  so that the dynamical  constraint is satisfied also at $t=0$.
After these preliminaries we transform the dKP equation using 
(\ref{implicit0}), to find an equation for $F(\xi,y,t)$. 
\begin{proposition}
The equations (\ref{implicit0}) give  a solution to  the dKP equation 
with smooth initial data $u_0(x,y)$  in implicit form, if the function 
$F(\xi,y,t)$ satisfies  the equation
\begin{equation}
\label{eqF}
\left(\frac{F_{t} + tF_{y}^{2}}{1+tF_{\xi}}\right)_{\xi}= F_{yy},
\end{equation}
with initial data
\begin{equation}
\label{initial1}
F(x,y,0)=u_0(x,y).
\end{equation}
\end{proposition}
\begin{proof}
Differentiating the second equation in (\ref{implicit0}) with
respect to $x$, $t$ and $y$  we find \beq
\label{trans}
\xi_x=\dfrac{1}{\Delta},\quad \xi_t=-\dfrac{F+tF_t}{\Delta},\quad \xi_y=\dfrac{tF_y}{\Delta}
\eeq
where we have defined $\Delta = 1 + tF_{\xi}$. Thus the derivatives of $u$ 
with respect to the variables are 
\beq
\label{ut}
u_t = F_{\xi}\xi_t + F_t = \frac{F_t-FF_{\xi}}{\Delta}, 
\eeq
and 
\begin{equation}
\label{uxy}
u_x=\dfrac{F_\xi}{\Delta},\quad u_{y}=\dfrac{F_y}{\Delta}. 
\end{equation}

Now the Hopf equation becomes 
\beq
0 = u_t + uu_x = \dfrac{F_t}{\Delta}, 
\label{Hopf_tr}
\eeq
which confirms that $F$ is time-independent in this case. Differentiating
(\ref{uxy}) a second time, we find
\[
u_{yy} = \left(\frac{F_y}{\Delta}\right)_y - 
\left(\frac{F_y}{\Delta}\right)_{\xi} 
\frac{t F_y}{\Delta} = \frac{1}{\Delta}\left[
F_{yy} - \left(\frac{t F_y^2}{\Delta}\right)_{\xi}\right],
\]
after some manipulations. But this means that if $u(x,y,t)$ satisfies 
the dKP equation (\ref{dKP}), $F(\xi,y,t)$ satisfies  (\ref{eqF}) 
with initial condition (\ref{initial1}). 
\end{proof}

We rewrite the equation (\ref{eqF}) in the evolutionary form
\[
F_t=\partial_{\xi}^{-1}F_{yy}+t(\partial_{\xi}^{-1}F_{yy} F_{\xi}-F_y^2)
\]
where $ \partial_{\xi}^{-1}$ is the inverse of a derivation. We observe from the above equation that the nonlinear terms
are multiplied by the time $t$ and this show that for small times the nonlinear effects are damped.
This observation  qualitatively  explains the fact that the function $F(\xi,y,t)$ develops a  singularity after $u(x,y,t)$ becomes singular.

For the remainder of this paper
we will focus on solutions to the transformed equation (\ref{eqF}). 
We observe that (\ref{eqF}) also conserves the integrals over 
$F$ and $F^2$, which we will use to check our numerics. Namely 
for $n$ integer one has 
\[
\int_{\mathbb{R}^2}u^n dx dy=\int_{\mathbb{R}^2}F^n \Delta d\xi dy = 
\int_{\mathbb{R}^2}F^n d\xi dy + 
\frac{t}{n+1}\int_{\mathbb{R}^2}(F^{n+1})_{\xi} d\xi dy = 
\int_{\mathbb{R}^2}F^n d\xi dy.
\]
In particular, conservation of the  $L^2$  norm (\ref{mass})  gives the  
constraint
\beq
\label{L2}
\int_{\mathbb{R}^2}F^2(\xi,y,t)d\xi dy=\int_{\mathbb{R}^2}u_0^2(x,y,0)dx dy.
\eeq

The transformation (\ref{implicit0}) has been constructed so as to 
unfold the overturned profile onto the initial condition in the case
of a $y$-independent initial condition. It is thus intuitive 
that if the overturning is modulated in the $y$-direction, it is 
unfolded onto a function $F(\xi,y,t)$ which shows no overturning, 
and having a weak dependence on $y$ and $t$ only.

\section{Overturning of the profile}
\label{sec:overturning}
\begin{figure}
\centering
   \includegraphics[width=0.49\textwidth]{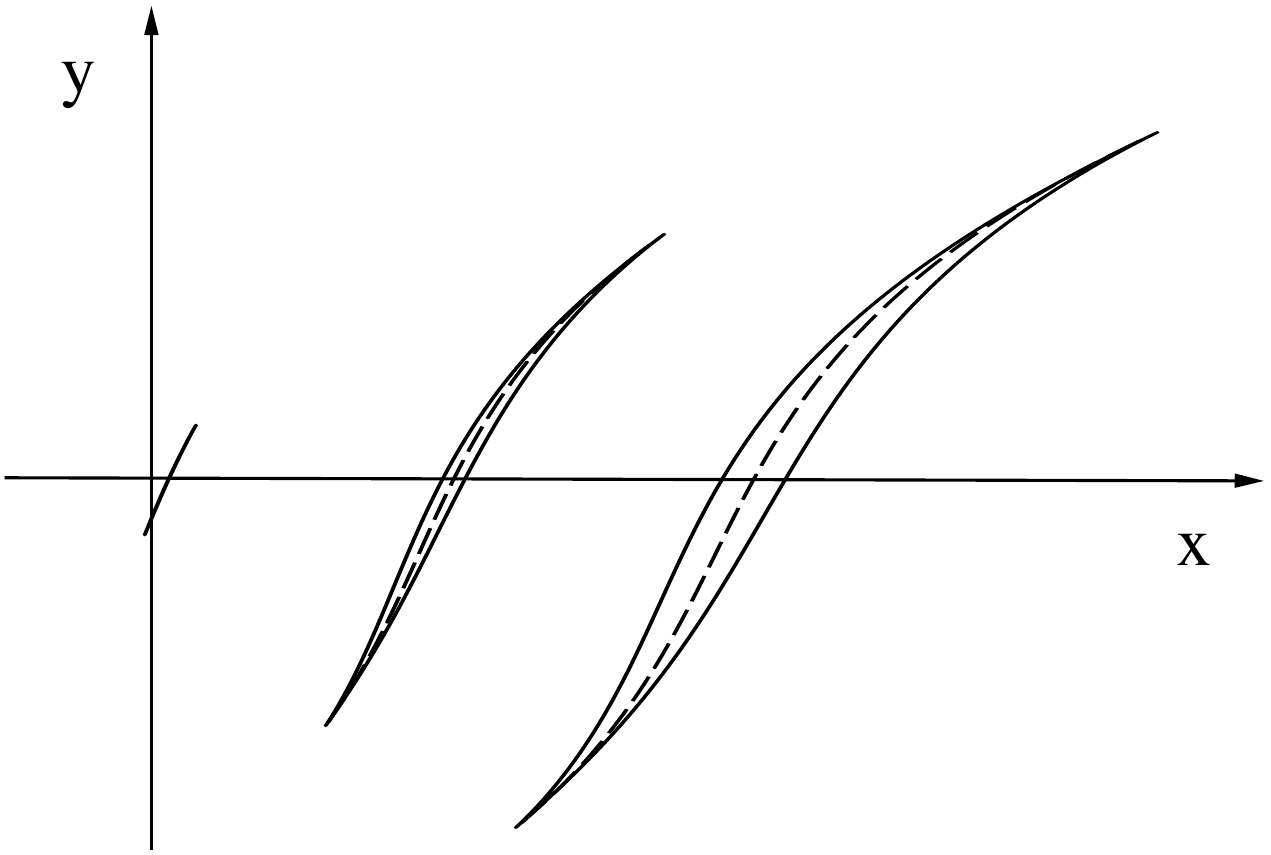} 
  \caption{A typical sequence of wave breaking, as described by 
(\ref{eqS}), showing the lip-shaped domain inside which the wave overturns.
The singularity first appears at the origin, then spreads
rapidly in the direction (\ref{transversal}). The scale of the
lip is $\bar{t}^{3/2}$ in the $x$-direction, and $\bar{t}^{1/2}$ 
in the $y$-direction. Full lines are solutions of 
$\Delta = 0$ at $\bar{t}= 0.01, 0.1$, and $0.4$, while the dashed line
is the shock front, which has to be inserted in accordance with the 
shock condition (\ref{shock_cond_final}), to be discussed in 
Section~\ref{sec:shock} below. 
}
 \label{fig:lip}
   \end{figure}
For generic initial data the solution of the dKP equation 
encounters a gradient catastrophe at points  where the transformation
(\ref{implicit0}) is not invertible \cite{MS08}
\begin{equation}
\label{eqS}
\Delta(\xi,y,t)\equiv 1+tF_{\xi}(\xi,y,t)=0,
\end{equation}
and  consequently  the gradients $u_x$ and $u_y$ go to infinity, cf.
(\ref{uxy}). This is illustrated in Fig.~\ref{fig:lip} for 
generic initial data, based on the local description to be developed below.
The singular time $t_c$ where the gradient catastrophe occurs first is 
the smallest $t$ such that (\ref{eqS}) holds. Since for $t<t_c$ the 
quantity $\Delta(\xi,y,t)$ has a definite sign in the $\xi$ and $y$ plane, 
$\Delta(\xi,y,t_c)$ must be a zero as well as an extremum:
$\Delta = \Delta_{\xi} = \Delta_y = 0$. Thus the two-dimensional 
gradient catastrophe is characterized by the equations:
\begin{equation}
\label{GC}
\begin{split}
&1+tF_{\xi}(\xi,y,t)=0\\
&F_{\xi\xi}=0\\
&F_{\xi y}=0\\
&u(x,y,t)=F(\xi,y,t)\\
&x=tF(\xi,y,t)+\xi.
\end{split}
\end{equation}
The first three equations of (\ref{GC}) determine the coordinates
$\xi_c,y_c$, and $t_c$ of the singularity in transformed variables,
taken as the origin in Fig.~\ref{fig:lip}.
The $x$ and $u$ coordinates are recovered by substitution into 
the last two equations. One finds that 
$(tF_y\partial_x+\partial_y)u(x,y,t)=F_y<\infty$, hence there is 
no gradient catastrophe in the transversal direction characterized 
by the vector field
\beq
tF_y {\bf e}_x + {\bf e}_y,
\label{transversal}
\eeq
see Fig.~\ref{fig:lip}. 
For generic  initial conditions, the second derivatives of $\Delta$
will be nonzero at the gradient catastrophe:
\beq
F_{\xi\xi\xi}(\xi_c,y_c,t_c)\neq 0,\quad F_{\xi\xi y }(\xi_c,y_c,t_c)\neq 0 
\quad F_{\xi yy}(\xi_c,y_c,t_c)\neq 0.
\label{generic}
\eeq
The conditions (\ref{GC}),(\ref{generic}) correspond to a cusp 
singularity in the notation of \cite{Alinhac95}, and will be found to 
describe the generic singularity for the dKP solution.
The condition that $F$ remains smooth, and thus the right hand side
$F_{yy}$ of (\ref{eqF}) is finite, results in the additional constraints
\begin{equation}
\label{constr}
F_t^c+t_c(F_y^c)^2=0,\quad F_{\xi t}^c=0,\quad F_{ty}^c+2t_cF^c_{yy}F_y^c=0,
\end{equation} 
where with a super-script we indicate the derivatives evaluated 
at the critical point.

We now give a local description of the two-dimensional wave 
front $u(x,y,t)$, based on expanding $F(\xi,y,t)$ near the 
gradient catastrophe described by (\ref{GC}). Our numerical 
simulations confirm that $F(\xi,y,t)$  remains smooth in the $(\xi,y)$ plane  not only near the 
first singularity, but well beyond. The region where the 
wave is multivalued has the typical lip shape also seen in 
the caustic surface of light waves near the cusp catastrophe
\cite{Nye99}. We will show them to be self-similar with 
width $\bar{t}^{3/2}$ in the horizontal direction and $\bar{t}^{1/2}$ in 
the transversal direction where $\bar{t}=t-t_c$, as done in \cite{MS08}. 
The same scalings have been observed in \cite{PLGG08} in the context of 
the 2-dimensional kinematic wave equation.

In order to illustrate the way in which (\ref{implicit0}) unfolds 
the singularity, it is instructive to consider a family of exact
solutions to (\ref{dKP}) obtained in \cite{MS11}:
\beq
u(x,y,t)=\dfrac{1}{\sqrt{t}}B\left(x-\dfrac{y^2}{4t}-2ut\right),
\label{exact_dKP}
\eeq
where $B$ is an arbitrary function of one variable. The validity of
(\ref{exact_dKP}) can be checked explicitly by substitution. Clearly 
(\ref{exact_dKP}) can be re-parameterized in the form 
\beq
u(x,y,t)=\dfrac{1}{\sqrt{t}}B\left(\zeta-\dfrac{y^2}{4t}\right),\quad 
x=2\sqrt{t}B(\zeta-\dfrac{y^2}{4t})+\zeta, 
\label{trans_MS}
\eeq
which shows that $\tilde{F}$, as defined by (\ref{F_MS}), is
\[
\tilde{F}(\zeta,y,t) = \dfrac{1}{\sqrt{t}}B(\zeta-\dfrac{y^2}{4t}).
\]

A singularity in the dKP solution occurs when the second 
equation in (\ref{trans_MS}) is no longer invertible; the first 
time this occurs is the critical time $t_c>0$, determined by
\[
\sqrt{t_c}=\min_{\zeta\in \mathbb{R}}\left(-\dfrac{1}{2B_\zeta}\right).
\]
In order to write (\ref{exact_dKP}) in terms of our 
function $F(\xi,y,t)$, we use the double re-parameterisation:
\begin{align}
\label{transfA}
&u(x,y,t)=F(\xi,y,t),\quad x=tF(\xi,y,t)+\xi,\\
\label{transfB}
& F(\xi,y,t)=\dfrac{1}{\sqrt{t}} B(\zeta-\dfrac{y^2}{4t}),\quad 
\xi=\sqrt{t}B(\zeta-\dfrac{y^2}{4t})+\zeta.
\end{align}
We observe that $F(\xi,y,t)$ has a singularity when the second 
transformation in (\ref{transfB}) is no longer invertible, namely at 
a critical time for the transformed equation (\ref{eqF}) 
\[
t_c^{(F)} = 4 t_c. 
\]
 It is also straightforward to check 
that at $t_c$, $F(\xi,y,t)$ satisfies the constraints (\ref{constr}).
Indeed one calculates directly from (\ref{transfB}) that 
\[
F_t=-\dfrac{B}{2t^{\frac{3}{2}}}\dfrac{2\sqrt{t}B'+1}{\sqrt{t}B'+1} + 
\dfrac{y^2B'}{4t^{\frac{5}{2}} }\dfrac{1}{\sqrt{t}B'+1},\quad 
F_y=-\dfrac{yB'}{2t^{\frac{3}{2}}(\sqrt{t}B'+1)},
\]
so at the critical time one obtains  the relations 
\[
F_t^c=-\dfrac{y_c^2}{4t^3_c},\quad F_y^c=\dfrac{y_c}{2t_c^2},
\]
which satisfy the first of the constraints in (\ref{constr}); the
remaining  constraints (\ref{constr}) are checked analogously. 

\subsection{Local analysis}
\label{sub:local}
In order to study the solution near the gradient catastrophe we 
expand the generalized characteristic equation (\ref{implicit0}) in
a Taylor series near $t_c$, $x_c,$ $y_c,$ $u_c$ and $\xi_c$.  Part of the analysis below is already contained in \cite{MS08}, \cite{MS12}.
 Introducing variables relative to the singularity
as
\[
\bar{x}:=x-x_c\quad \bar{t}:=t-t_c,\quad\bar{y}:= y-y_c,
\quad \bar{\xi}:\xi-\xi_c,
\]
we have argued that $\bar{x}\sim \bar{t}^{3/2}$ and 
$\bar{y}\sim \bar{t}^{1/2}$. Since $\Delta u\sim\bar{t}^{1/2}$, it 
follows from the first equation of (\ref{implicit0}) that 
$\bar{\xi}\sim\bar{t}^{1/2}$. Thus to be consistent, we include 
all terms up to $O(\bar{t}^{3/2})$:
\begin{equation}
\begin{split}
\label{char1}
\bar{x}&=\bar{t}(F^c+t_cF_t^c)+\bar{t}\bar{y}(F_y^c+t_cF^c_{yt}) +
t_c\left(F^c_{y}\bar{y}+\dfrac{1}{2}F^c_{yy}\bar{y}^2 +
\dfrac{1}{6}F_{yyy}^c\bar{y}^3\right) \\
&+\dfrac{t_c}{6}F^c_{\xi\xi\xi}\bar{\xi}^3+
\dfrac{t_c}{2}F^c_{\xi\xi y}\bar{y}\bar{\xi}^2+
\left(\frac{t_c}{2}F_{\xi yy}^c\bar{y}^2 + F_{\xi}^c\bar{t}\right)\bar{\xi}
+o(\bar{t}^2,\bar{y}^4,\bar{\xi}^4,
\bar{t}(\bar{y}^2+\bar{\xi}^2)).
\end{split}
\end{equation}

This suggests introducing the shifted variables 
(using $t_c = -1/F_{\xi}^c$):
\begin{equation}
\label{XT}
\begin{split}
&\zeta=F^c_\xi\left(\bar{\xi}+
\dfrac{F^c_{\xi\xi y}}{F^c_{\xi\xi\xi}}\bar{y}\right)\\
&X=\frac{1}{k}\left[\bar{x}-\bar{t}(F^c+t_cF_t^c)-
\bar{t}\bar{y}(F_y^c+t_cF^c_{yt})-t_c\left(F^c_{y}\bar{y}+
\dfrac{1}{2}F^c_{yy}\bar{y}^2 + \dfrac{1}{6}F_{yyy}^c\bar{y}^3\right)\right.\\
&\left.-\frac{1}{3}t_c\dfrac{(F^c_{\xi\xi y})^3}{(F^c_{\xi\xi\xi})^2}\bar{y}^3+
\frac{1}{2}t_c\dfrac{F^c_{\xi\xi y}F^c_{\xi y y}}{F^c_{\xi\xi\xi}}\bar{y}^3+
F^c_\xi\dfrac{F^c_{\xi\xi y}}{F^c_{\xi\xi\xi}}\bar{y}\bar{t}\right]\\
&T=\frac{1}{k}\left[\bar{t}+\frac{t^2_c}{2}\bar{y}^2
\left(\dfrac{(F^c_{\xi\xi y})^2}{F^c_{\xi\xi\xi}}-F^c_{\xi yy}\right)\right],\\
&k=\dfrac{t_c^4 F^c_{\xi\xi\xi}}{6}
\end{split}
\end{equation} 
so that in the variable $\zeta$, (\ref{char1}) takes the form
\begin{equation}
\label{s_orig}
-\zeta^3 + T\zeta=
X+o(\bar{t}^2,\bar{y}^4,\bar{\xi}^4,\bar{t}(\bar{y}^2+\bar{\xi}^2)).
\end{equation}

\begin{figure}
\centering
   \includegraphics[width=0.7\textwidth]{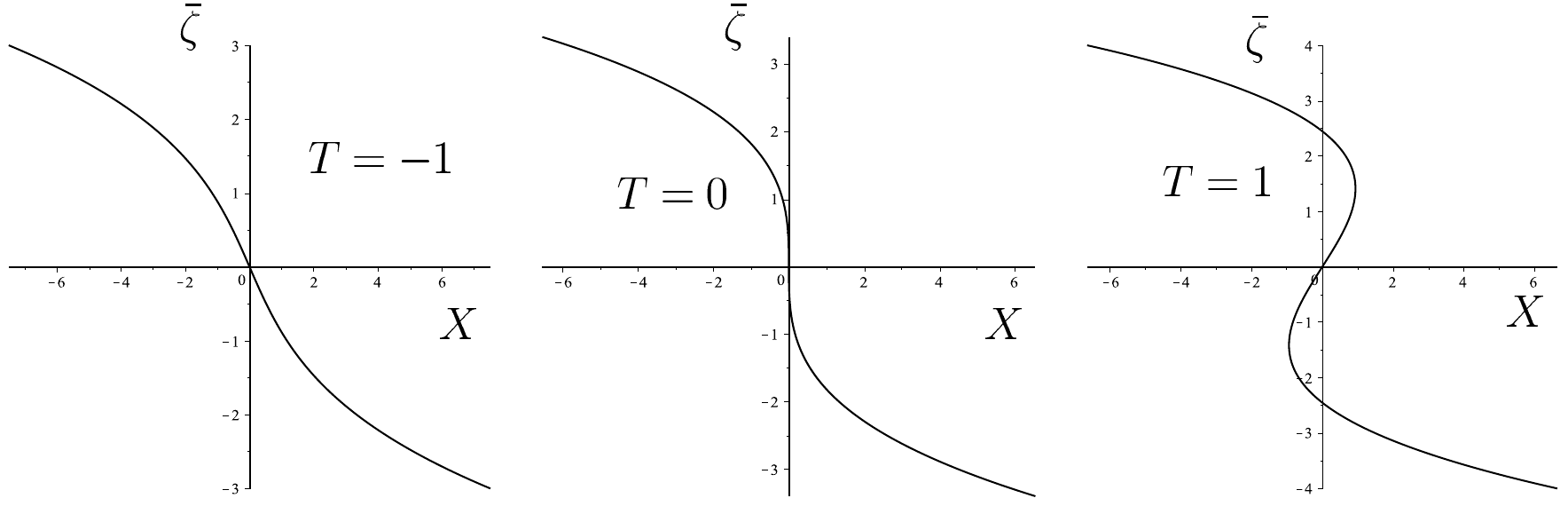} 
  \caption{The universal s-curve described by 
(\ref{char2}); for $T>0$ the profile turns over to form a 
multivalued region. 
}
 \label{fig:s}
   \end{figure}
Using the estimates $\bar{\xi}\sim\bar{y}\sim\bar{t}^{1/2}$ and 
$\bar{x}\sim\bar{t}^{3/2}$ identified previously, the scaling
\begin{equation}
\label{scaling_lambda}
\begin{split}
&X\to \lambda X\\
&\bar{t}\to\lambda^{\frac{2}{3}} \bar{t}\\
&\bar{y}\to\lambda^{\frac{1}{3}}\bar{y}\\
&\zeta\to\lambda^{\frac{1}{3}}\zeta,
\end{split}
\end{equation}
in the limit $\lambda\to 0$
reduces (\ref{s_orig}) to the 
universal s-curve
\begin{equation}
\label{char2}
-\zeta^3+T\zeta=X
\end{equation}
shown in Fig.~\ref{fig:s}. It is easy to confirm that 
the function $\zeta(X,T)$ defined by (\ref{char2}) solves 
\begin{equation}
\label{Hopf}
\zeta_T + \zeta\zeta_{X}=0
\end{equation}
with initial condition $\zeta(X,T=0)=(-X)^{\frac{1}{3}} $, 
completing our task of reducing (\ref{dKP}) locally to the Hopf
equation. A gradient catastrophe is encountered for 
$X=0$, $T=0$, and $\zeta=0$.

Using the identities (\ref{constr}), we can now calculate the 
solution to the dKP equation (\ref{dKP}), valid near the singularity. 
To leading order in the limit $\lambda \rightarrow 0$, it is consistent 
to expand $u(x,y,t)$ to linear order in $\bar{\xi},\bar{y}$:
\beq
u(x,y,t)-u_c = F(\xi,y,t) - F^c \simeq F_{\xi}^c\bar{\xi}+F_y^c\bar{y}=
\zeta(X,T) + \bar{\beta}\bar{y},
\label{u_exp}
\eeq
with 
\begin{equation}
\label{beta}
\bar{\beta} = F_y^c - \frac{F_{\xi}^cF_{\xi\xi y}^c}{F_{\xi\xi\xi}^c}.
\end{equation}

Thus putting $\bar{u}\equiv u(x,y,t)-u_c$, from (\ref{char2}) 
we find the local profile to be an s-curve, which has the 
universal similarity form: 
\begin{equation}
\label{s_u}
-\left(\bar{u}-\bar{\beta}\bar{y}\right)^3 + 
T\left(\bar{u}-\bar{\beta}\bar{y}\right) = X. 
\end{equation}
This is the central result of our theoretical analysis; the
formula (\ref{s_u}) is the unfolding of an $A_2$ singularity.
It is a complete description of the self-similar  behavior of the 
dKP solution near its  singularity for generic initial data.
In the $y$-independent case, (\ref{s_u}) coincides with the usual 
result (\ref{S}). We now derive the form of this multivalued valued 
region in the $x,y$-plane, shown previously in Fig.~\ref{fig:lip}. 

\subsection{Multivalued region}
As seen in Fig.~\ref{fig:s}, the function $\zeta=\zeta(X,T)$,
described by the cubic equation (\ref{char2}), becomes multivalued for $T>0$. 
From ${\displaystyle \frac{\partial X}{\partial \zeta}=0}$ it follows
that 
\begin{equation}
\label{Txi}
T = 3\zeta^2, \quad \mbox{or}\quad
\zeta=\pm\sqrt{T/3},
\end{equation}
so that for $X$ in the interval
\[
-\dfrac{2}{3\sqrt{3}} T^{3/2}\leq X\leq \dfrac{2}{3\sqrt{3}} T^{3/2},
\]
the function $\zeta(X,T)$ is multivalued. 

Reversing the coordinate
transformations (\ref{XT}), we can write the first equation 
(\ref{Txi}) in the form
\begin{equation}
\label{t}
\bar{t}=\dfrac{1}{2}t_c\alpha\left(\bar{\xi}^2 + 
2\beta\bar{y}\bar{\xi}+\gamma\bar{y}^2\right),
 \end{equation}
where we have introduced the constants
\begin{equation}
\label{alpha}
\alpha=t_cF^c_{\xi\xi\xi},\quad \beta=\dfrac{F^c_{\xi\xi y}}{F^c_{\xi\xi\xi}},
\quad \gamma=\dfrac{F_{\xi yy}^c}{F^c_{\xi\xi\xi}}.
\end{equation}
Alternatively, (\ref{t}) could also have been derived from 
(\ref{eqS}), and expanding $F$ in a power series around the singular
point. 

The $\bar{x}$-coordinate of the boundary of overturning can be found from 
(\ref{char1}), which using (\ref{t}) can be simplified to yield 
\eqa
\label{x}
&& \bar{x}=-\alpha\left(\dfrac{1}{3}\bar{\xi}^3 + 
\dfrac{\beta}{2}\bar{y}\bar{\xi}^2\right) + 
\bar{t}\left(F^c-t^2_c(F_y^c)^2\right) + \nonumber \\ 
&& \bar{t}\bar{y}\left(F_y^c-2t_c^2F^c_{yy}F_y^c\right) + 
t_c\left(F^c_{y}\bar{y} + \dfrac{1}{2}F^c_{yy}\bar{y}^2 + 
\dfrac{1}{6}F_{yyy}^c\bar{y}^3\right).
\eeqa
Equations (\ref{t}) and (\ref{x}) describe a curve in the 
$(\bar{x},\bar{y})$ plane, parameterized by $\bar{\xi}$. 
An example was shown previously in Fig.~\ref{fig:lip} for several 
values of $\bar{t}=t-t_c$, showing its characteristic ``lip'' shape
\cite{Arnold84}.

The overturned region starts from the singular point and then expands, 
as seen in Fig.~\ref{fig:lip}. To understand the scaling of this 
expansion, we introduce the independent variables
\begin{equation}
\label{X1Y1}
\left\{
\begin{array}{ll}
X_1&=\left[\bar{x}-\bar{t}(F^c-t^2_c(F_y^c)^2)-
t_c(F^c_{y}\bar{y}+\dfrac{1}{2}F^c_{yy}\bar{y}^2)\right]\bar{t}^{-3/2}\\
Y_1&=\bar{y}\bar{t}^{-1/2}.
\end{array}\right. 
\end{equation}
Then the lip described by equations (\ref{t}) and (\ref{x}) 
is reduced to the time-independent similarity form:
\begin{equation}
\label{slitX1Y1}
\left\{
\begin{array}{ll}
&\dfrac{1}{2}t_c\alpha\left(s^2+2\beta Y_1s+\gamma Y_1^2\right)=1\\
&X_1=-\alpha\left(\dfrac{1}{3}s^3+\dfrac{1}{2}\beta s^2 Y_1\right) + 
\delta_1 Y_1 + \dfrac{\delta_2}{6} Y_1^3,
\end{array}\right.
\end{equation}
with the additional constants 
\begin{equation}
\label{delta}
\delta_1=F_y^c-2t_c^2F^c_{yy}F_y^c, \quad
\delta_2 = t_c F_{yyy}^c .
\end{equation}

\begin{figure}
\centering
   \includegraphics[width=0.49\textwidth]{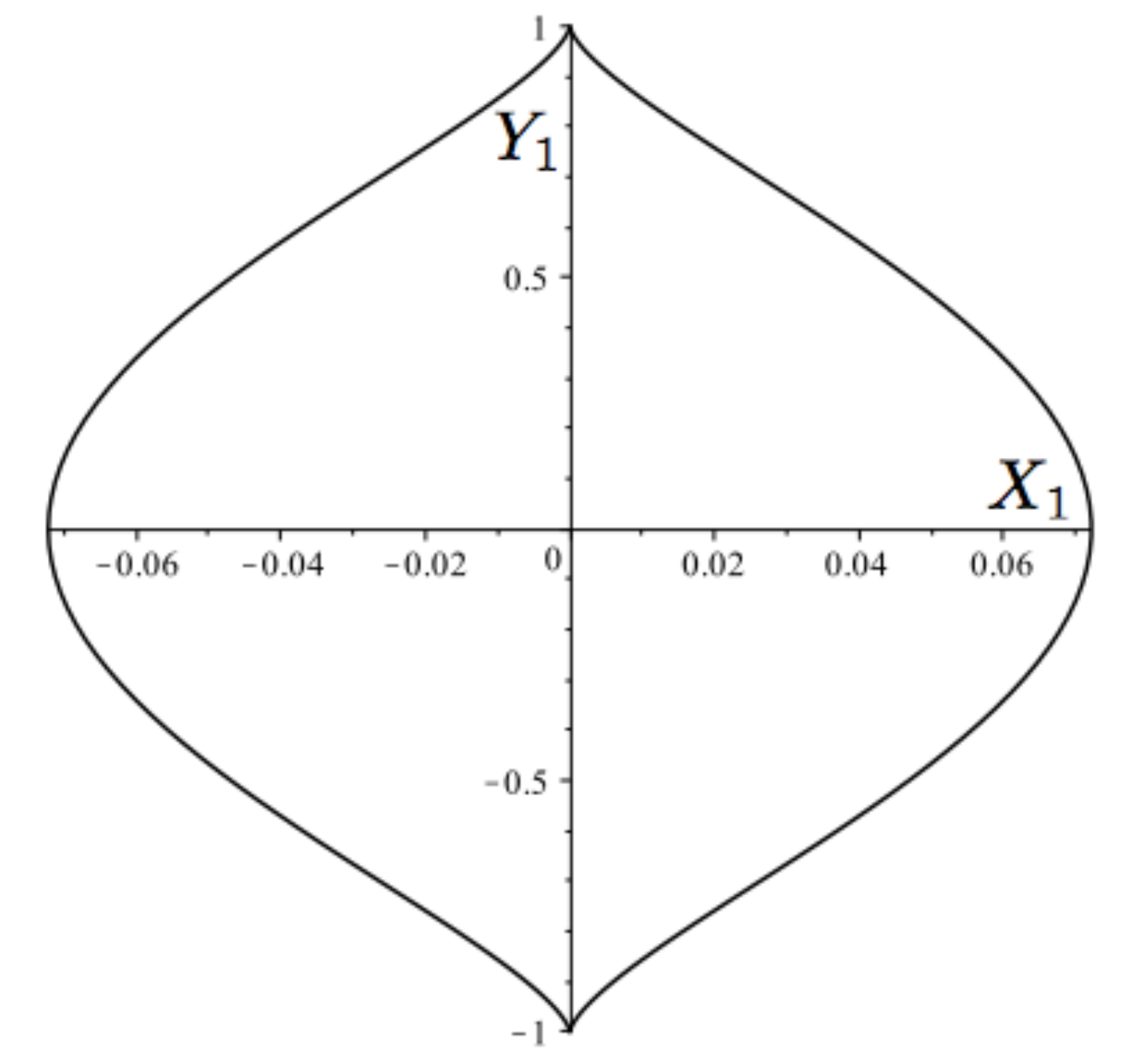} 
  \caption{The symmetric lip according to (\ref{lip_symm}), ending
in a cusp. 
}
 \label{fig:lip_symm}
   \end{figure}

This demonstrates that the gradient catastrophe in the dKP equation 
has a universal spatial signature, parameterized by the constants
$\alpha,\beta,\gamma,\delta_1$, and $\delta_2$, all of which can be 
computed in terms of the initial data and its derivatives at the 
point of gradient catastrophe $(x_c,t_c,y_c)$. The scalings introduced 
in (\ref{X1Y1}) imply that the lip expands as $\bar{t}^{3/2}$ in 
the propagation direction, and as $\bar{t}^{1/2}$ in the 
transversal direction, as announced previously. A characteristic 
feature is the cusp at the corner of the lip. This is seen most
easily for initial data which is even in $y$, for which the description 
simplifies considerably. All odd derivatives in $y$ vanish, and 
we obtain 
\begin{equation}
\label{lip_symm}
\left\{
\begin{array}{lll}
&&\dfrac{1}{2}t_c\alpha\left(s^2+\gamma Y_1^2\right)=1\\
&&X_1=-\dfrac{\alpha}{3}s^3,
\end{array}\right.
\end{equation}
shown in Fig.~\ref{fig:lip_symm}. Analyzing the neighborhood
of the point $s=0$, one finds directly that 
\begin{equation}
\label{lip_cusp}
X_1 = \pm\frac{\alpha}{3}(2\bar{Y}_1)^{3/2}\left(\bar{Y}_1-Y_1\right)^{3/2},
\quad \bar{Y}_1 = \sqrt{\frac{2}{t_c\alpha}},
\end{equation}
which is a generic 3/2 cusp \cite{EF13}. 

\section{Dissipative dKP equation and shock solutions}
\label{sec:shock}
The solution (\ref{s_u}) constructed in the previous subsection is 
unphysical for $\bar{t}>0$, in that it does not assign a unique 
value of $u$ to every point $x,y$ in the plane. In principle, 
one can construct an infinity of single-valued solutions from it,
by choosing different points at which to jump from one branch to the 
other. For conservation laws in many space dimensions, 
physically motivated constraints, known as generalized Rankine-Hugoniot 
jump conditions, have been introduced. As result, a weak solution of the 
equations (usually called the inviscid shock), is singled out 
uniquely \cite{Majda84,K70}.

Another way to select a unique solution after the singularity, is to consider
a dissipative version of (\ref{dKP}) with a viscous term added to
it, which keeps the solution regular at all times. In the limit 
of vanishing viscosity $\epsilon$ these regular solutions are expected to 
converge to (\ref{s_u}), with a particular jump condition being selected. 
In this case the shock is called the viscous shock. In the field of 
hyperbolic equations the problem of showing that the inviscid shock is 
equal to the viscous shock has generated a huge literature. We only mention
some important references in one dimension \cite{Bressan2005,Goodman92}, 
and many space dimensions \cite{Metivier2008}. Below we give an heuristic 
derivation of the equivalence of the inviscid and viscous shock for the 
dKP equation.

We consider the dissipative form of the dKP equation 
\begin{equation}
\label{DKP}
(u_t+uu_x-\epsilon(u_{xx} +c y_{yy}))_x=u_{yy},
\quad u(x,y,t=0,\epsilon))=u_0(x,y),
\end{equation}
with $c \geq 0$,
which satisfies 
\begin{equation}
\label{diss}
\frac{1}{2}\frac{\partial}{\partial t}\int_{\mathbb{R}^2} u^2(x,y,t)dxdy=
-\epsilon\left(u_x^2 +c u_y^2\right) < 0. 
\end{equation}
For given $\epsilon$-independent initial data, 
the solution  $u(x,y,t,\epsilon)$ of the dissipative equation (\ref{DKP})
is expected to be approximated as $\epsilon\to 0$ and $t<t_c$ by the 
solution $u(x,y,t)$ of the dKP equation (\ref{dKP}).

\subsection{Shock position}
On the other hand, (\ref{p}) is still satisfied at finite $\epsilon$, 
so (smooth) solutions of (\ref{DKP}) still conserve $u$ in the limit 
$\epsilon\to 0$. From the condition that $u$ be satisfied across a shock, 
we can use the generalized Rankine-Hugoniot jump conditions as in 
\cite{Majda84}, which determines the shock position 
(see also \cite{Chorin_Marsden}). 
Namely, if $v_n$ is the normal velocity of the shock, one obtains 
\beq
v_n(u_1-u_2) = 
\left({\bf f}\cdot{\bf n}\right)_1 - \left({\bf f}\cdot{\bf n}\right)_2,
\label{shock_cond_n}
\eeq
where ${\bf n}$ is the normal to the shock front, and indices 1 and 2 
denote values in front and in the back of the shock, respectively. 
Assuming that the shock position is given by the curve $x_s(\bar{y},\bar{t})$,
and using the flux ${\bf f}$ from (\ref{cons1}), this yields
\beq
\dot{x}_s(u_1-u_2) = \frac{u_1^2 - u_2^2}{2} + 
\frac{\partial x_s}{\partial y}\int_{x_2}^{x_1} u_y dx,
\label{shock_cond_x}
\eeq
where $x_{1/2}$ are $x$-values approaching the shock from the front 
and from behind, respectively. 

Now the singular contribution to $u$ 
across the shock can be written in  the form 
\[
u  = u_2+(u_1-u_2)\theta(\bar{x} - x_s(\bar{y},\bar{t})),
\]
where $\theta(x)$ is the Heaviside function and $u_{1,2}$ become 
functions only of $y$ and $t$ on the shock front 
$\bar{x}=x_s(\bar{y},\bar{t})$.
Hence
\[
u_y =(u_2)_y+(u_1-u_2)_y \theta(\bar{x} - x_s(\bar{y},\bar{t})) - 
(u_1-u_2)\frac{\partial x_s}{\partial y}
\delta(\bar{x} - x_s(\bar{y},\bar{t})),
\]
and from (\ref{shock_cond_x}) the jump condition at the 
shock finally becomes 
\beq
\dot{x}_s = \frac{u_1 + u_2}{2} - \left(\frac{\partial x_s}{\partial y}
\right)^2.
\label{shock_cond_final}
\eeq
Note that the shock speed in the $x$-direction is not only 
an average between $u$-values in front and in the back of the 
shock as for the Hopf equation, but on account of the right hand 
side of (\ref{dKP}) an additional term arises. 

Since we have mapped (\ref{dKP}) locally to the Hopf equation (\ref{Hopf}),
standard theory \cite{Chorin_Marsden} tells us that the shock 
should be at $X = 0$, according to (\ref{XT}) the equation for the 
front becomes
\eqa
&& x_s(\bar{y},\bar{t}) = \bar{t}(F^c+t_cF_t^c) + 
\bar{t}\bar{y}(F_y^c+t_cF^c_{yt}) + 
t_c\left(F^c_{y}\bar{y}+\dfrac{1}{2}F^c_{yy}\bar{y}^2 + 
\dfrac{1}{6}F_{yyy}^c\bar{y}^3\right) + \nonumber \\
&& \frac{1}{3}t_c\dfrac{(F^c_{\xi\xi y})^3}{(F^c_{\xi\xi\xi})^2}\bar{y}^3 -
\frac{1}{2}t_c\dfrac{F^c_{\xi\xi y}F^c_{\xi y y}}{F^c_{\xi\xi\xi}}\bar{y}^3 -
F^c_\xi\dfrac{F^c_{\xi\xi y}}{F^c_{\xi\xi\xi}}\bar{y}\bar{t}. 
\label{front_s}
\eeqa
This equation indeed satisfies (\ref{shock_cond_final}) to leading order, 
since 
\beq
\dot{x}_s(\bar{y},\bar{t}) = F^c+t_cF_t^c + \bar{y}t_cF^c_{yt}
+ \bar{\beta}\bar{y}
\label{xst}
\eeq
and 
\beq
\left(\frac{\partial x_s}{\partial y}\right)^2 = 
\left[t_c\left(F^c_{y} + F^c_{yy}\bar{y}\right) + O(\bar{t})\right]^2
= -t_cF_t^c - t_cF_{ty}^c\bar{y} + O(\bar{t}),
\label{xsy}
\eeq
having used (\ref{constr}). 
On the other hand, $\zeta_{\pm} = \pm\sqrt{T}$ at 
$X=0$, and so according to (\ref{u_exp}) 
\[
\frac{u_1 + u_2}{2} = F^c + \bar{\beta}\bar{y}.
\]
Combining the last three equations one can see that 
the approximate shock front (\ref{front_s})
satisfies (\ref{shock_cond_final}) to leading order. In 
Fig.~\ref{fig:lip} we have plotted  (\ref{front_s}) as the dashed line. 

\subsection{Shock structure}
Having found the shock position, we now investigate the inner structure
of the shock, in case a small amount of viscosity is present. This is 
achieved by mapping (\ref{DKP}) onto Burgers' equation \cite{Whithambook},
which in addition to (\ref{Hopf}) contains a dissipative contribution. 
We are looking for a solution $u(x,y,t;\epsilon)$ of the dissipative dKP 
equation near the gradient catastrophe $(x_c,y_c,t_c)$ of the 
(inviscid) dKP equation. To this end we use the ansatz
\beq
u(x,y,t;\epsilon)= u_c+h(X,T;\epsilon)+\bar{y}\bar{\beta},
\label{multiscale}
\eeq
with $X$ and $T$ defined in (\ref{XT}). Using the same scalings as
before, and balancing $u_t \propto \epsilon u_{xx}$, we are 
led to the multiscale expansion 
\[
h(X,T;\epsilon)=\lambda^{\frac{1}{3}}H({\cal X},{\cal T};\varepsilon)
+O(\lambda^{\alpha}),\quad \alpha >\frac{1}{3},\]
\begin{equation}
\label{subs}
X=\lambda{ \cal X},\quad T=\lambda^{\frac{2}{3}}{ \cal T},\quad 
\epsilon=\lambda^{\frac{4}{3}}\varepsilon,\quad 
\bar{y}=\lambda^{\frac{1}{3}}{\cal Y},
\end{equation}
and find the following theorem:
\begin{theorem}\label{Theo1}
Let $u(x,y,t;\epsilon)= u_c+h(X,T;\epsilon)+\bar{y}\bar{\beta}$
be a solution of the dissipative dKP equation (\ref{DKP}) with 
$X$ and $T$ defined in (\ref{XT}).
Suppose that  for  $|t-t_c|$ small the limit
\[
H({ \cal X}, {\cal T};\varepsilon)=
\lim_{\lambda\to 0}\lambda^{-\frac{1}{3}}h(\lambda{ \cal X},\,
\lambda^{2/3}{\cal T};\lambda^{\frac{4}{3}}\varepsilon)\]
exists and the function $H({ \cal X}, {\cal T};\varepsilon)$ satisfies the asymptotic conditions
\begin{equation}
\label{asymp}
H({ \cal X}, {\cal T};\varepsilon)=\mp|{\cal X}|^{\frac{1}{3}}\mp 
\dfrac{{\cal T}}{3}|{\cal X}|^{-\frac{1}{3}}+O(|{\cal X}|^{-\frac{5}{3}}),
\quad |{\cal X}|\to\infty
\end{equation}
for each fixed ${\cal T}\in\mathbb{R}$.
Then the function $H({\cal X},{\cal T};\varepsilon)$ satisfies 
the Burgers equation
\begin{equation}
\label{lim2}
H_{\cal{T}}+HH_{\cal{X}}=\sigma
H_{\cal{X}\cal{X}},\quad \sigma=\dfrac{\varepsilon}{k} 
\left(1 +c (t_cF_y^c)^2\right)
\end{equation}
with $k$ defined in (\ref{XT}).
\end{theorem}
\begin{proof}
Inserting (\ref{multiscale}) into the dissipative dKP equation one obtains
\begin{multline}
(H_{\cal{T}}+HH_{\cal{X}}-\dfrac{\varepsilon}{k} \left(1 + (t_cF_y^c)^2\right)
H_{\cal{X}\cal{X}})_{{\cal X}}+
\lambda^{-\frac{1}{3}}H_{{\cal X}{\cal X}}\left( \dfrac{\partial X}{\partial t}-
\left(\dfrac{\partial X}{\partial y}\right)^2 + 
F_c+\bar{y}(F^c_y-F^c_\xi\dfrac{F^c_{\xi\xi y}}{F^c_{\xi\xi\xi}})\right)\\
=H_{\cal{T}\cal{X}}\dfrac{\partial T}{\partial y}\dfrac{\partial X}{\partial y}+
\lambda^{\frac{1}{3}}H_{\cal{T}\cal{T}}
\left(\dfrac{\partial T}{\partial y}\right)^2+
\lambda  k H_{{\cal T}}\dfrac{\partial^2 T}{\partial y^2} + 
\lambda^{2/3}kH_{\cal{X}}\dfrac{\partial^2  X}{\partial y^2} - 
\varepsilon\left(H_{{\cal X}{\cal X}}
\left(\left(\dfrac{\partial X}{\partial y}\right)^2-
(t_cF_y^c)^2\right)\right)\\
-\frac{\varepsilon}{k}\left(\lambda^{1/3} H_{\cal{T}\cal{X}}
\dfrac{\partial T}{\partial y}\dfrac{\partial X}{\partial y}+
\lambda^{2/3}H_{\cal{T}\cal{T}}\left(\dfrac{\partial T}{\partial y}\right)^2+
\lambda^{4/3}k H_{{\cal T}}\dfrac{\partial^2  T}{\partial y^2}+
\lambda k H_{\cal{X}}\dfrac{\partial^2  X}{\partial y^2}\right)_{{\cal X}}.
\end{multline}
Using  (\ref{XT}), the constraints (\ref{constr}), and the substitution 
(\ref{subs}) one arrives at the relation
\[
(H_{\cal{T}}+HH_{\cal{X}}-\frac{\varepsilon}{k} \left(1 + c(t_cF_y^c)^2\right)
H_{\cal{X}\cal{X}})_{{\cal X}}=O(\lambda^{\frac{1}{3}}),
\]
which in the limit $\lambda\to 0$ shows that the derivative of 
(\ref{lim2}) is equal to zero.
 In order to fix the integration constant  we use the asymptotic condition (\ref{asymp}).
\end{proof}

We remark that the asymptotic condition (\ref{asymp}) implies that the  local solution near the point of  singularity formation, matches the outer solution 
given by (\ref{char2}) and (\ref{u_exp}).
We conclude that near the gradient catastrophe, up to the constant term $u_c$ 
as well as a term linear in $y$, in a suitable co-ordinate system
the solution to the dissipative dKP equation reduces to the solution of 
the one-dimensional Burgers equation. We will argue below that the 
particular solution to the Burgers equations relevant near the critical 
point, and described by the asymptotic form (\ref{asymp}), also satisfies 
the equation
\begin{equation}
\label{eqPearcey}
{\cal X}=H{\cal T}-H^3+6\sigma HH_{{\cal X}}-4\sigma^2 H_{{\cal X}{\cal X}}.
\end{equation}
\subsubsection{Burgers equation.}
To find the local solution near the shock, let us recall the solution 
to Burgers' equation 
\begin{equation}
\label{Burgers}
v_t+vv_x= \nu v_{xx},
\end{equation} 
with initial data $v_0(x)$, where 
$\nu $ is a positive constant. 
An exact solution is obtained via the Cole-Hopf transformation 
\cite{Hopf54},\cite{Whithambook} to give the formula:
\beq
v(x,t,\nu)=-2\nu\partial_x\log\int_{-\infty}^{\infty}
e^{-\frac{G(\eta,x,t)}{2\nu}}d\eta,
\label{CH}
\eeq
where 
\beq
G(\eta,x,t)=\int_0^{\eta}v_0(s)ds+\dfrac{(x-\eta)^2}{2t}. 
\label{CHkernel}
\eeq

For $\nu\to 0$ the leading contributions to the integral come from 
the neighborhood of the critical points of $G$, namely
\begin{equation}
\label{characteristics}
\partial_{\eta} G(\eta,x,t)|_{\eta=\xi}=v_0(\xi)-\dfrac{x-\xi}{t}=0.
\end{equation}
Let us assume first that there is only one such critical point, 
which means that using the method of the steepest descent \cite{Whithambook}, 
the integral can be approximated as 
\[
\int_{-\infty}^{\infty}e^{-\frac{G(\eta,x,t)}{2\nu}}d\eta \approx
\sqrt{\frac{4\pi\nu}{\partial_{\xi}^2G(\xi,x,t)}}e^{-\frac{G(\xi,x,t)}{2\nu}},
\]
where $\xi=\xi(x,t)$ is a solution of (\ref{characteristics}). Direct
evaluation of (\ref{CH}), using the characteristic 
condition (\ref{characteristics}), then yields the solution 
\beq
v(x,t,\nu)=v_0(\xi) + 
\nu\dfrac{v_0''(\xi)t}{(v_0'(\xi)t+1)^2}+O(\nu^2),
\label{perturbation}
\eeq
whose leading order contribution in the limit $\nu\rightarrow 0$ 
is the solution of the Hopf equation by characteristics. In addition,
(\ref{perturbation}), contains a linear correction coming from the 
viscosity. Alternatively, the term linear in $\nu$ can 
also be obtained using perturbation theory.

The approximation (\ref{perturbation}) remains valid as long as the 
function $G(\eta,x,t)$ has an isolated generic critical point, before the
appearance of a gradient catastrophe. However, after the critical 
triple point $(x_c,t_c)$ of the Hopf equation, where $v_0'(\xi_c)t_c+1=0$ and 
$v_0''(\xi_c)=0$, (\ref{characteristics}) has {\it three} solutions,
as illustrated on the left of Fig.~\ref{FPearcey} below. 
Near this point $G(\eta,x,t)$  can be expanded in a Taylor series as
\[
\Delta G:=G(\eta,x,t)- G(\xi_c,c_c,t_c)\simeq v_0'''(\xi_c)
\frac{\bar{\eta}^4}{4!}-\bar{\eta}^2\frac{\bar{t}}{2t_c^2}-
\bar{\eta}\dfrac{\bar{x}-v_c\bar{t}}{t_c}+v_c(\bar{x}-v_c\bar{t})+
v_c^2\bar{t},
\]
where $\bar{x}=x-x_c$, $\bar{t}=t-t_c$, 
$\bar{\eta}=\eta-\xi_c$, $v_c=v_0(\xi_c)$ and $v_0'''(\xi_c)>0$.
Thus near such critical point the solution of Burgers' equation 
can in the limit $\nu\to 0$  be approximated by 
\begin{equation}
\label{Pearcey1}
v(x,t,\nu)\simeq v_c - 2\nu\partial_x\log
\int\limits_{-\infty}^{\infty}\exp\left[-\frac{1}{2\nu}
\left(v_0'''(\xi_c)\frac{\bar{\eta}^4}{4!}-\bar{t}\frac{\bar{\eta}^2}{2t_c^2}-
\frac{\bar{\eta}}{t_c}(\bar{x}-v_c\bar{t})\right)\right]d\bar{\eta}.
\end{equation}
Some rescaling leads to the following  (see also \cite{DubEl})
\begin{theorem}\cite{Ilin}\label{Theo2}
Near a gradient catastrophe $(x_c,t_c)$ for the solution of the Hopf 
equation $v_t+vv_x=0$, the solution $v(x,t,\nu)$ 
of (\ref{Burgers}) admits the following expansion
 \begin{equation}
v(x,t,\nu) = v_c + \left(\frac{\nu}{\kappa}\right)^{1/4}U
\left(\frac{\bar{x}-v_c\bar{t}}{\left(\kappa\nu^3\right)^{1/4}},
\frac{\bar{t}}{\left(\kappa\nu\right)^{1/2}}\right) + 
O(\nu^{1/2}),
 \end{equation}
where $v_c=v(x_c,t_c)$,  $\kappa=t_c^4v_0'''(\xi_c)/6$
and the function 
U=$U(a,b)$ is defined by
\begin{equation}
\label{Pearcey2}
U(a,b)=-2\partial_a\log\int_{-\infty}^{+\infty}e^{-\frac{1}{8}(z^4 -2z^2b+4za)}dz.
\end{equation} 
\end{theorem}
\begin{figure}
\centering
   \includegraphics[width=0.7\textwidth]{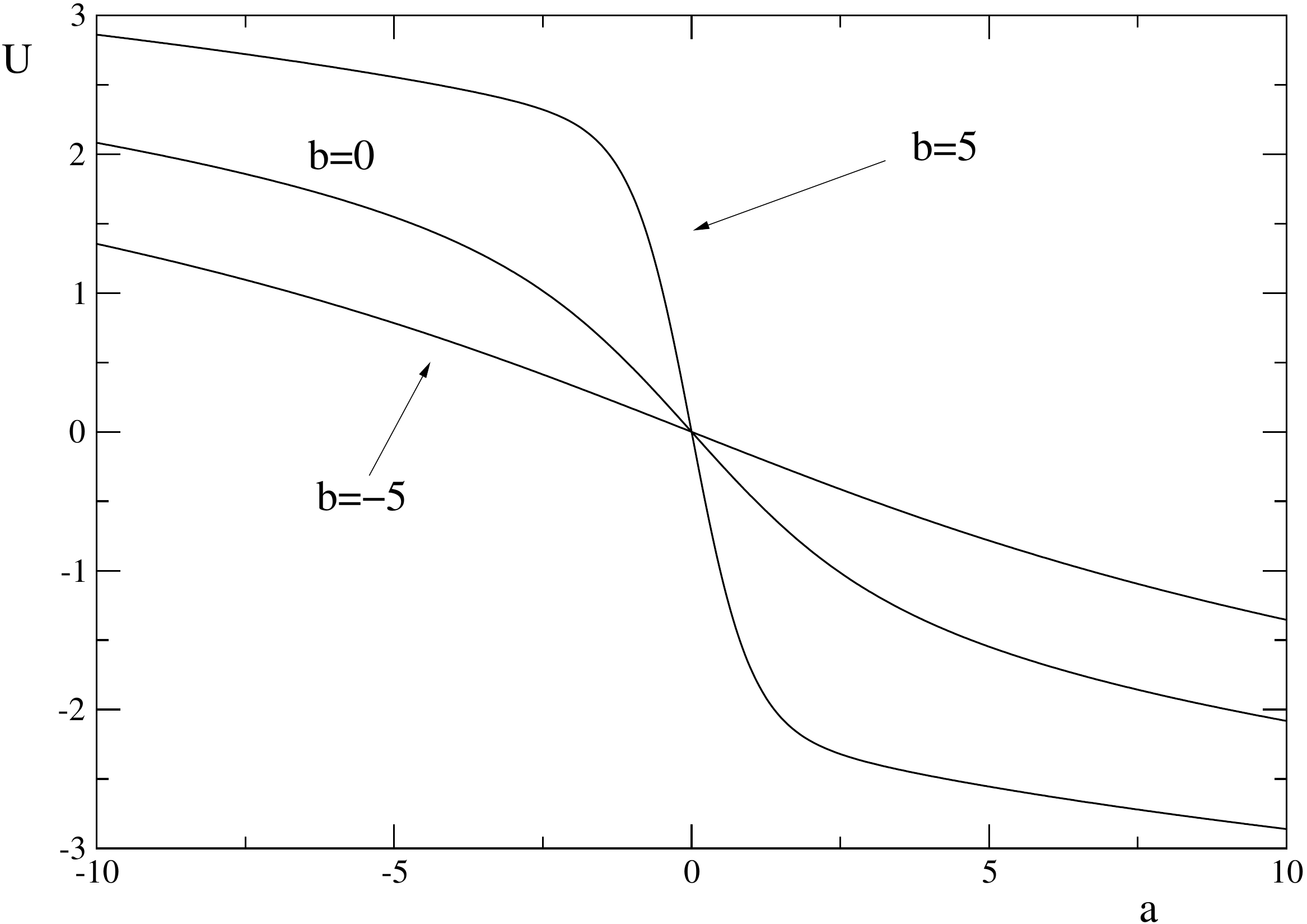} 
  \caption{The Pearcey function $U(a,b)$, for three different values of $b$.
}
 \label{fig:U}
   \end{figure}
\begin{remark}\label{moro}
The function $U(a,b)$ satisfies both the Burgers equation
\[
U_b+UU_a=U_{aa}
\]
and the non-linear ODE in the $a$-variable, containing $b$ as a 
parameter \cite{Moro13}
\beq
a=Ub-U^3+6UU_a-4U_{aa}.
\label{U_rel}
\eeq

The behavior of $U(a,b)$ is illustrated in Fig.~\ref{fig:U} for 
negative, positive, and vanishing values of reduced time $b$, 
performing the integral in (\ref{Pearcey2}) numerically. 
For large $|a|$ and fixed $b$ the  integral (\ref{Pearcey2}) behaves as 
the root of the cubic equation (\ref{char2}) (see below)
\[
U(a,b)=\mp|a|^{\frac{1}{3}}\mp \dfrac{b}{3}|a|^{-\frac{1}{3}}+
O(|a|^{-\frac{5}{3}}),\quad |a|\to\infty
\]
\end{remark}

The integral in (\ref{Pearcey2}) is related to the standard Pearcey function 
\cite{NIST:DLMF}, which describes the diffraction pattern near a cusp caustic
\cite{Nye99}, by a complex rotation. The relation (\ref{U_rel}) 
is convenient in deducing the asymptotic properties of $U(a,b)$; it 
follows from 
\[
\int_{-\infty}^{\infty}\frac{d}{dz}e^{-\frac{1}{8}(z^4 -2z^2b+4za)}dz = 0. 
\]

\begin{figure}
\centering
   \includegraphics[width=0.7\textwidth]{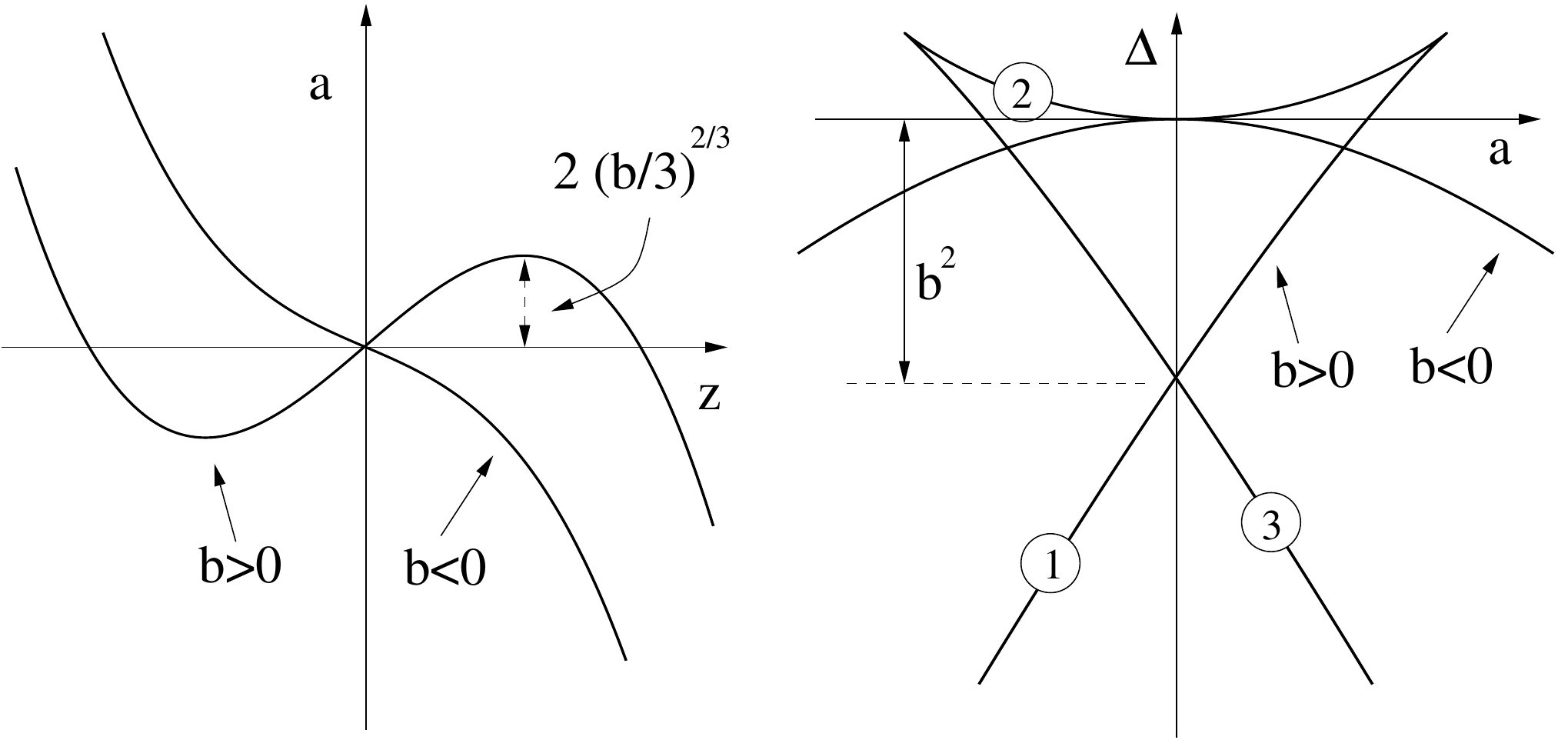} 
  \caption{The critical contributions to the integral (\ref{Pearcey2})
near a cusp catastrophe, at constant reduced time $b$. 
On the left, the critical points; there is 
a unique solution for $b<0$, and three solutions for 
$\left|a\right| \le 2(b/3)^{3/2}$ if $b>0$. On the 
right, the argument $\Delta$ of the exponential; for $b<0$ there is a
single contribution, for $b>0$ there are three contributions to a given 
value of $a$. }
 \label{FPearcey}
   \end{figure}
In catastrophe theory \cite{PS78} the potential 
\beq
\Delta(z) = z^4 -2z^2b+4za 
\label{potential}
\eeq
(the weight in the exponent of (\ref{Pearcey2}))
is the standard unfolding of the cusp catastrophe, which is a co-dimension 
2 singularity. For $b<0$ (before the gradient catastrophe), there is 
only one critical point 
\beq
0 = \frac{d\Delta}{dz} = 4z^3 - 4zb + 4a ,
\label{potential_crit}
\eeq
which is the case we considered before (see Fig.~\ref{FPearcey}). 
Evaluating $\Delta$ at the critical point (\ref{potential_crit})
yields 
\beq
\Delta = -3z^4 + 2z^2b, \quad a = -z^3 + zb.
\label{Delta_swallow}
\eeq
For $b<0$ this gives the single-valued curve shown on the right of 
Fig.~\ref{FPearcey}, which leads to the solution (\ref{perturbation}). 

If on the other hand $b>0$ (after the gradient catastrophe), in the 
range $\left|a\right| \le 2(b/3)^{3/2}$ there are 
three critical points. Thus the integral (\ref{Pearcey2}) has three 
contributions, with different values of $\Delta$ (cf. (\ref{Delta_swallow})), 
which lie on a swallowtail figure, as shown on the right of 
Fig.~\ref{FPearcey}. The integral is dominated by the {\it smallest}
value of $\Delta$, as long as the solutions are well separated. This means
we must have $b\gg 1$ (cf. Fig.~\ref{FPearcey}), or 
$\bar{t}/\epsilon^{1/2}\gg 1$. Closer to gradient catastrophe, 
a more sophisticated asymptotics is needed, or one has to evaluate
the integral numerically, as we will do below. However, outside
of the region $b\lesssim 1$, the integral is dominated by either 
solution $z_1$ or $z_3$. The changeover occurs for $a=0$, where 
$\Delta(z_1) = \Delta(z_3)$, namely on the line 
$\bar{x}-v_c\bar{t}=0.$ This is exactly the shock front near the 
gradient catastrophe $(x_c,t_c)$.
 \subsubsection{Pearcey integral and dissipative dKP equation.}
Choosing  $\lambda=\epsilon^{\frac{3}{4}}$ in Theorem~\ref{Theo1} we obtain 
that the solution  to the dissipative dKP equation satisfies in the 
rescaled variables (\ref{subs})  the Burgers equation (\ref{lim2}) with 
$\varepsilon=1$. Furthermore for $t<t_c$ such solution is asymptotic to 
the Hopf solution (\ref{char2}). Combining these observations with 
Theorem~\ref{Theo2} and remark~\ref{moro},   we come up the following  conjecture.

\begin{conjecture}
Let us consider the   double scaling limit $\epsilon\to 0$,  $x\to x_c$, $y\to y_c$  and $t\to t_c$  in such a way that  the ratios
\[
\frac{X}{\epsilon^{3/4}},\quad \frac{T}{\epsilon^{1/2}},
\]
remain bounded with  $X$ and $T$  defined in (\ref{XT}).
Then the solution $u(x,y,t;\epsilon)$ of the dissipative dKP equation 
near  the first singularity  for the solution of the dKP equation 
is described by  the expansion
 \begin{equation}
\label{KPPearcy}
u(x,y,t;\epsilon)\simeq u_c+\sigma^{1/4}
U\left(\frac{X}{\sigma^{3/4}},
\frac{T}{\sigma^{1/2}}\right)
+\bar{y}\bar{\beta}
+ O(\epsilon^{1/2}),
\end{equation}
where 
\[
\sigma=\epsilon\dfrac{6\left(1 +c \left(t_cF_y^c\right)^2\right)}
{F^c_{\xi\xi\xi}t_c^4},
\]
 and the 
function $U(a,b)$ is the Pearcey integral defined in (\ref{Pearcey2}).
\end{conjecture}
For $y$-symmetric initial data the expression (\ref{KPPearcy}) reduces 
to the form
\begin{equation}
\label{KPPeracy}
u(x,y,t;\epsilon)\simeq u_c+\sigma^{1/4}
U\left(\frac{\bar{x}-u_c\bar{t}-t_cF^c_{yy}\bar{y}^2/2}
{k\sigma^{3/4}},
\frac{\bar{t}-t^2_cF^c_{\xi yy}\bar{y}^2/2}
{k\sigma^{1/2}}\right)
+O(\epsilon^{1/2}),
\end{equation}
with $k$ defined in (\ref{XT}).
The center of the (smooth) shock front is located at $X=0$, 
as found previously in the inviscid limit. 

\section{Numerical solution}
\label{sec:num}
In this section we present numerical solutions of the transformed
version (\ref{eqF}) of the dKP equation, which remain smooth well 
beyond the gradient catastrophe of the original equation (\ref{dKP}),
as we will demonstrate below. In addition, we treat the dissipative 
dKP equation (\ref{DKP}), whose solutions are also observed to remain 
smooth. We use a Fourier method for the spatial dependence, and an 
\emph{exponential time differencing} (ETD) scheme for the time 
dependence, as previously for the dKP equation \cite{KR13}.

Both equations are written in \emph{evolutionary form}
\begin{equation}
    F_{t}=\partial_{\xi}^{-1}F_{yy}+t(F_{\xi}\partial_{\xi}^{-1}F_{yy}-F_{y}^{2}),
    \label{Fev}
\end{equation}
and 
\begin{equation}
    u_{t}+uu_{x}=\partial_{x}^{-1}u_{yy}+\epsilon\left(u_{xx} +c u_{yy}\right),
    \label{dKPdis}
\end{equation}
with a small dissipation parameter $\epsilon$. In Fourier space, 
the antiderivatives $\partial_{\xi}^{-1}$ and $\partial_{x}^{-1}$ are 
represented as Fourier multipliers $-i/k_{\xi}$ and $-i/k_{x}$, respectively. 
Here $k_{\xi}$, $k_{x}$, $k_{y}$ are the dual Fourier variables of $\xi$, 
$x$, $y$ respectively, and the Fourier transform of a variable will be 
denoted by a hat. Thus (\ref{Fev}) and (\ref{dKPdis}) can be written
in the form
\begin{equation}
    \hat{u}_{t}=\mathcal{L}\hat{u}+\mathcal{N}(\hat{u}),
    \label{uhat}
\end{equation}
where $\mathcal{L}$ is a linear, {\it diagonal} operator, which is 
$ik_{y}^{2}/k_{\xi}$ for (\ref{Fev}), and $ik_{y}^{2}/k_{x}-\epsilon k_{x}^{2}$ 
for (\ref{dKPdis}), and $\mathcal{N}(\hat{u})$ is a nonlinear term.
The idea of the ETD scheme to be used here is to treat the 
linear part of (\ref{uhat}) exactly. We use the fourth order EDT 
method by Cox and Matthews \cite{CM02}, but other schemes offer 
a very similar performance \cite{KR11}. 

To satisfy the constraint (\ref{constID}) on the initial condition, 
we choose initial data as the derivative of a function from the
Schwarz space of rapidly decreasing smooth functions. This is well 
suited to a Fourier method, since a Schwarz function can be continued as a 
smooth periodic function to within our finite numerical precision. 
However, the nonlocality of (\ref{Fev}) and (\ref{dKPdis}) 
implies that solutions will develop tails with an algebraic decrease towards 
infinity. This follows already from the Green function of the linearized 
equations \cite{KSM}. It was shown in \cite{KSM,KR11} that discontinuities 
at the boundaries of the computational domain can nevertheless be avoided 
by choosing a large enough domain, and one can achieve 
\emph{spectral accuracy} (an exponential decrease of the numerical error 
with the number of Fourier modes) over the time scales considered. 

The antiderivative in both (\ref{Fev}) and (\ref{dKPdis}) leads 
to Fourier multipliers which are singular in the limit of small wave 
numbers. 
These terms are regularized in Fourier space by adding a term of the order 
of the machine precision ($\sim10^{-16}$ here). In \cite{KSM}, the dKP 
equation (\ref{dKP}) was solved for $\partial_{x}^{-1}u$, which is possible 
since solutions maintain the property of being the derivative of a Schwarz
function. Together with an exponential integrator treating the term 
$ik_y^2/k_x$ explicitly, this addressed all numerical 
problems stemming from this singular operator. 

However, an explicit treatment of all singular terms is not 
possible for (\ref{Fev}), since $\mathcal{N}$ is singular as well,
which leads to numerical problems for $k_{\xi}\to0$. This can 
be addressed by applying a Krasny filter \cite{krasny}: all 
Fourier coefficients with modulus smaller than some threshold 
(typically $10^{-10}$) will be put equal to 0. In all cases considered,
our numerical algorithm could now be continued well beyond the first 
gradient catastrophe. For longer times, the above mentioned 
algebraic tails will lead to a slower decrease of the Fourier 
coefficients and thus to numerical problems once the numerical errors 
are of the order of the Krasny filter. For long time computations, 
which are beyond the scope of the current paper, one would have to use 
considerably larger domains and higher resolutions, or alternatively 
a spectral approach as in \cite{BK}.

The accuracy of the numerical solution to (\ref{eqF}) was monitored via 
the decrease of the Fourier coefficients, and checking the conservation 
of the $L^2$ norm (cf. (\ref{mass}),(\ref{L2})). To this end we compute
\begin{equation}
    \delta(t) = 1-\frac{M(t)}{M(0)},
    \label{delta_M}
\end{equation}
whose time dependence will be a measure of the numerical error. As 
shown in \cite{etna,KR13}, the maximum error in $F$ may well be 
one to two orders of magnitude greater than $\delta$,  but within these 
limits $\delta$ is nevertheless a reliable indicator of 
the accuracy, if the Fourier coefficients decrease sufficiently rapidly.

\subsection{Shock formation for symmetric initial data}
We begin with the simplest case of initial data symmetric with respect 
to $y\to-y$. We choose the same initial condition as \cite{KR13}, 
\begin{equation}
    u_{0}(x,y)=-6\partial_{x}\mbox{ sech}^{2}\sqrt{x^2+y^2},
    \label{u0sym}
\end{equation}
who solved the dKP equation (\ref{dKP}) in its original form. 
Near the gradient catastrophe, (\ref{dKP}) develops a discontinuity,
and the numerical scheme employed in \cite{KR13} breaks down. By contrast,
using the transformed equation (\ref{eqF}), we are able to reach the gradient 
catastrophe with much lower resolution (using serial instead of 
parallel computers), but are also able to continue the computation beyond 
the first and even secondary wave-breaking events. Beyond the gradient
catastrophe, we identify the lines $\Delta = 0$ along which the gradient 
of the solution blows up (cf. Fig.~\ref{fig:lip}), and show that the 
solution of (\ref{eqF}) yields the expected weak solution of dKP inside
the lip region. We also show that the solution of (\ref{eqF}) stays 
regular on time scales of order unity.

\begin{table}
 \centering
    \leavevmode
\begin{tabular}{|c|c| c |c|c|c|c|}
\hline
Breaking event & Initial data&$t_c$&$x_c$&$y_c$&$u_c$&$\xi_c$\\
\hline
First& $-6\partial_{x}\mbox{ sech}^{2}\sqrt{x^2+y^2}$&0.222&1.79& 0&2.543&1.227\\
\hline
Second & $-6\partial_{x}\mbox{ sech}^{2}\sqrt{x^2+y^2}$&0.300&-2.033 
&0&-2.48&-1.289\\
\hline
\end{tabular}
\caption{Critical parameters for the first two wave breaking events, with
symmetric initial data (\ref{u0sym}).}
\label{table:critical}
\end{table}

In \cite{KR13}, the first wave breaking event was observed at the critical 
time $t_{c}= 0.2216\dots$, see Table~\ref{table:critical}. Here we can 
identify $t_{c}$ directly from a solution of (\ref{eqF}) by tracing the 
minimum of $\Delta$ over space. The first time this quantity vanishes or 
becomes just negative will be taken as the time $t_{c}$. We use 
$N_{x}=N_{y}=2^{9}$ Fourier modes for $x,y\in[-5\pi,5\pi]^{2}$ 
and $N_{t}=1000$ time steps for $t\leq 0.23$. The first negative value 
of $\Delta$ is recorded for $t=0.222\ldots$, which is in agreement with 
\cite{KR13} to within the accuracy
of at least two digits. However, the present calculation 
requires much lower resolution to reach similar accuracy ($N_{x}=N_{y}=2^{9}$ 
compared to $N_{x}=N_{y}=2^{15}$ in \cite{KR13}), and accuracy can
easily be improved. For example, after determining the critical time to a 
certain accuracy, one uses the required resolution in time close to the 
previously determined $t_{c}$. This allows to determine the critical 
time  with the same precision as the solution to (\ref{eqF}), i.e., with the 
accuracy of the Krasny filter chosen here to be equal to $10^{-10}$. 
For our purposes an accuracy of the order of $10^{-3}$ will be sufficient. 

The location of the critical point was identified in \cite{KR13} as  
$x_{c} = 1.79\dots$ and $y_{c}=0$. Here it is calculated for $t=t_{c}$ by 
first finding the minimum $\xi_{c}=1.227\ldots$, $y_c=0$ of $\Delta$, where 
$F(\xi_{c},y_{c},t_{c})=2.543\ldots$. Then, using (\ref{GC}), we find
$x_{c}=\xi_{c}+t_{c}F(\xi_{c},y_{c},t_{c})=1.792\ldots$, again in 
excellent agreement with our previous result \cite{KR13}, estimated
to be correct to at least two digits. 

\begin{figure}
\centering
    \includegraphics[width=0.42\textwidth]{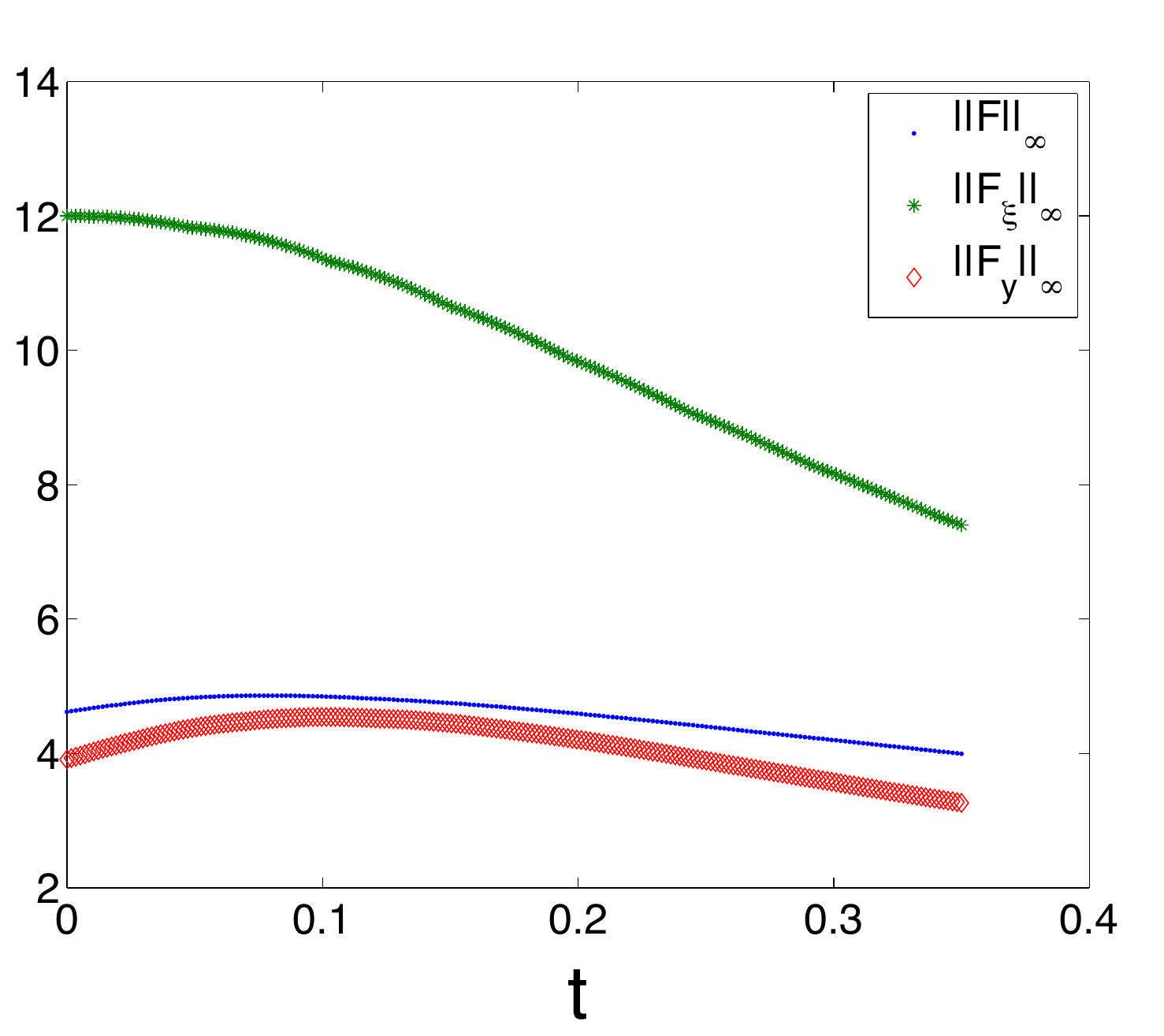}
    \includegraphics[width=0.49\textwidth]{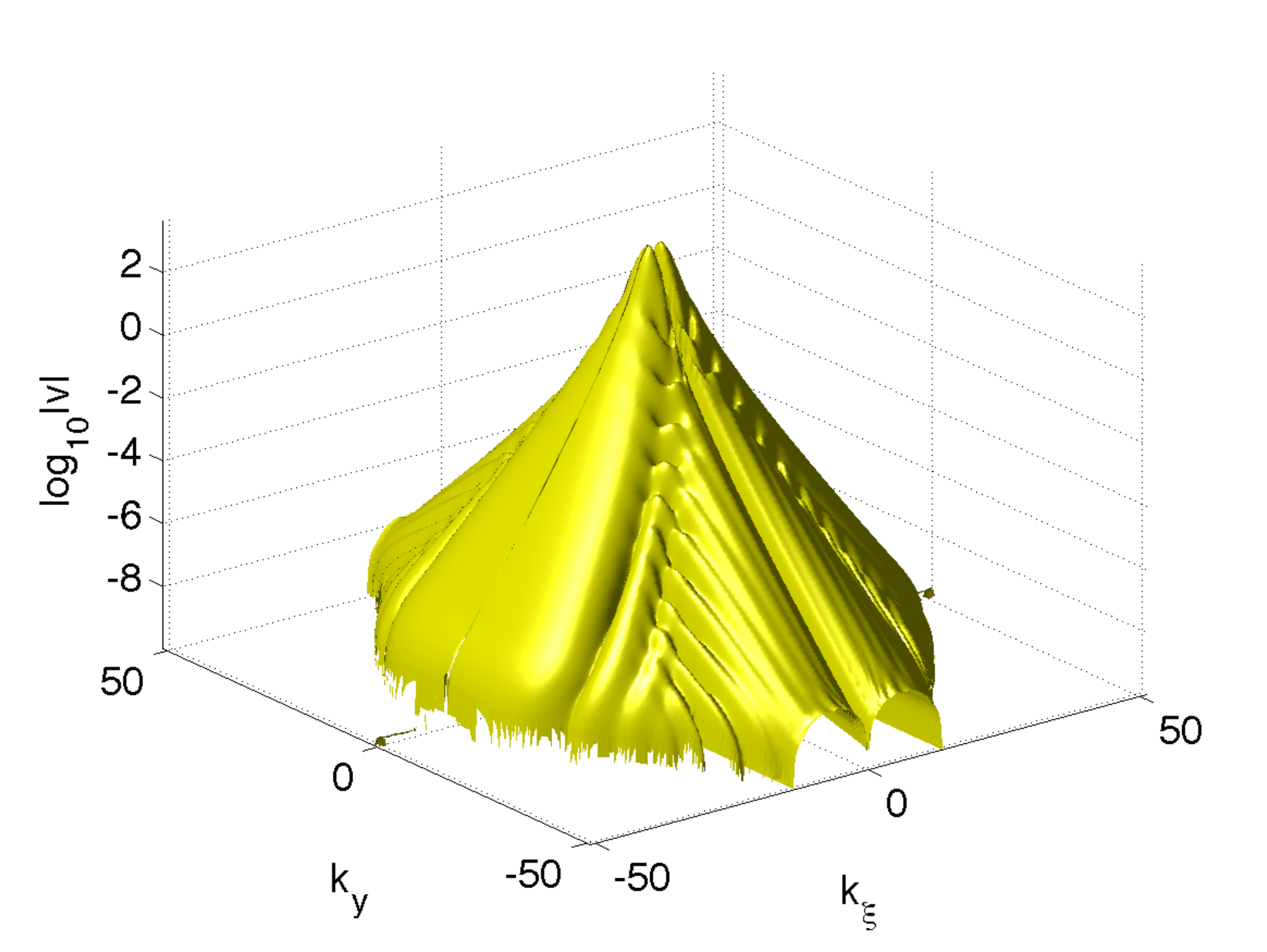}
 \caption{Measures of the smoothness of the solution to (\ref{eqF}) 
with initial data (\ref{u0sym}). On the left, the time dependence of 
the maximum norm of $F$, as well as of $F_{\xi}$ and $F_y$;
all decay for long times. 
On the right, the Fourier coefficients of the solution for $t=0.32$. }
 \label{Fsechnorm}
\end{figure}
However, the solution $F$ of (\ref{eqF}) stays perfectly regular 
well beyond the critical time $t_c$ of the dKP solution $u(x,y,t)$, 
as seen in Fig.~\ref{Fsechnorm}. On the left, we show that the
maximum norms of the first derivatives of $F$ remain bounded and smooth 
at $t_c$, and even decay for long times (of course, the derivatives of 
the original variable $u(x,y,t)$ diverge at a gradient catastrophe). 
On the right, for $t=0.32$ we demonstrate exponential decay of the Fourier 
coefficients to the level of the Krasny filter, as expected for a smooth 
function. The relative $L^2$ norm $\delta(t)$ (cf. (\ref{delta_M})) is 
conserved to the order of $10^{-14}$. On account of the algebraic decay 
of the solution in Fourier space, the computation cannot be run for much 
longer than $t = 0.35$ at the current resolution. To be able to do so 
using a Fourier method, larger domains and higher resolution would be 
needed. However, there is no indication that the solution of (\ref{eqF}) 
itself develops a singularity.

Thus it is possible to continue the computation beyond the first wave
breaking event, and to identify the second event, which occurs 
for negative $x$. This is of course not possible in the case of direct 
integration of (\ref{dKP}) as in \cite{KR13}, where the numerical method
fails at the first wave breaking. We use $N_{x}=2^{9}$, $N_{y}=2^{11}$ Fourier 
modes and $N_{t}=5000$ time steps for $t\leq 0.32$. Proceeding as for the 
first break-up in tracing the minimum of $\Delta(\xi,y,t)$,
we find $\tilde{t}_{c}=0.300\ldots$ and $\tilde{x}_{c}=-2.033\ldots$, see 
Table~\ref{table:critical}. 

\begin{figure}
\subfigure
{     \includegraphics[width=0.49\textwidth]{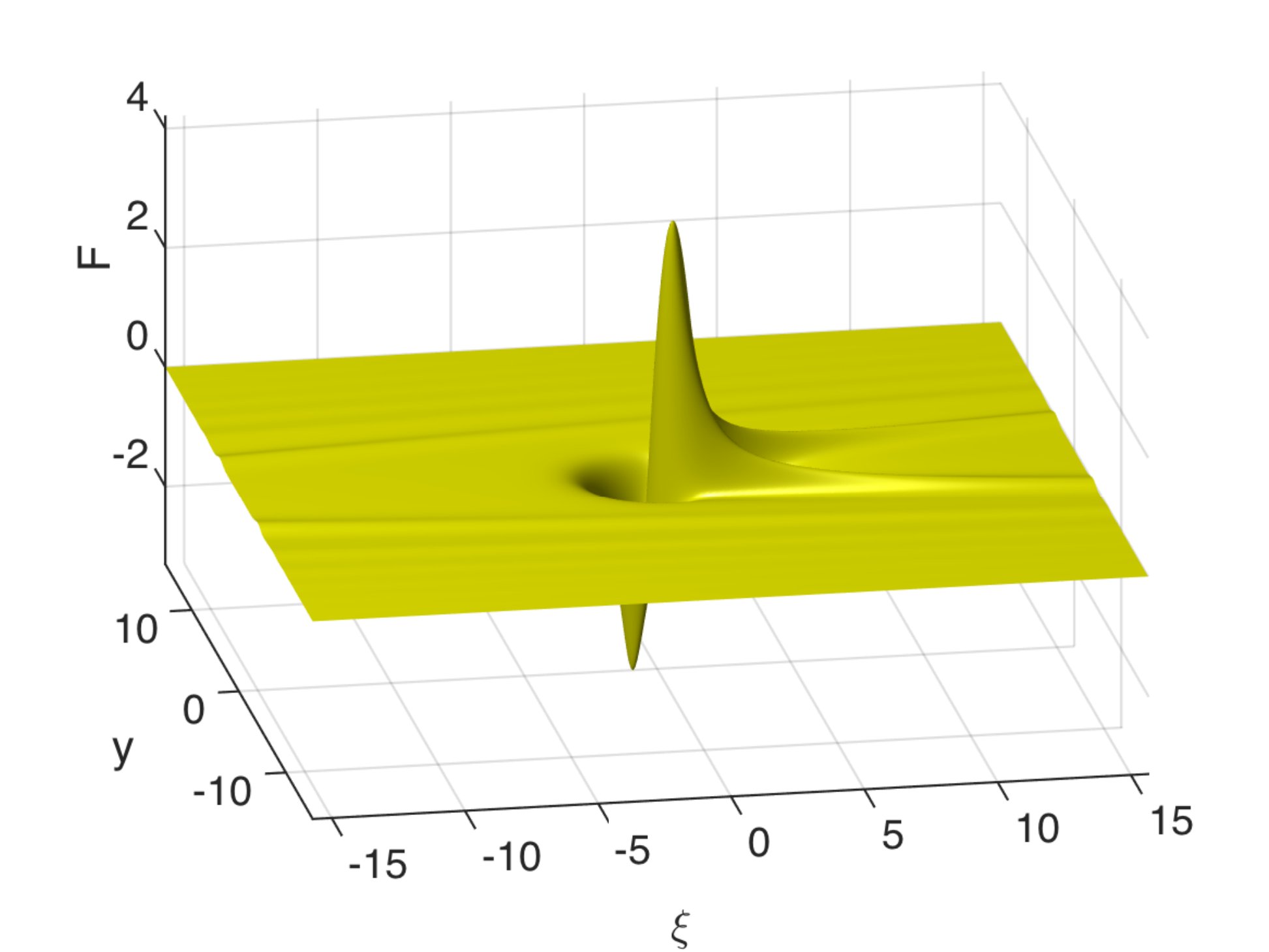}}
 \includegraphics[width=0.49\textwidth]{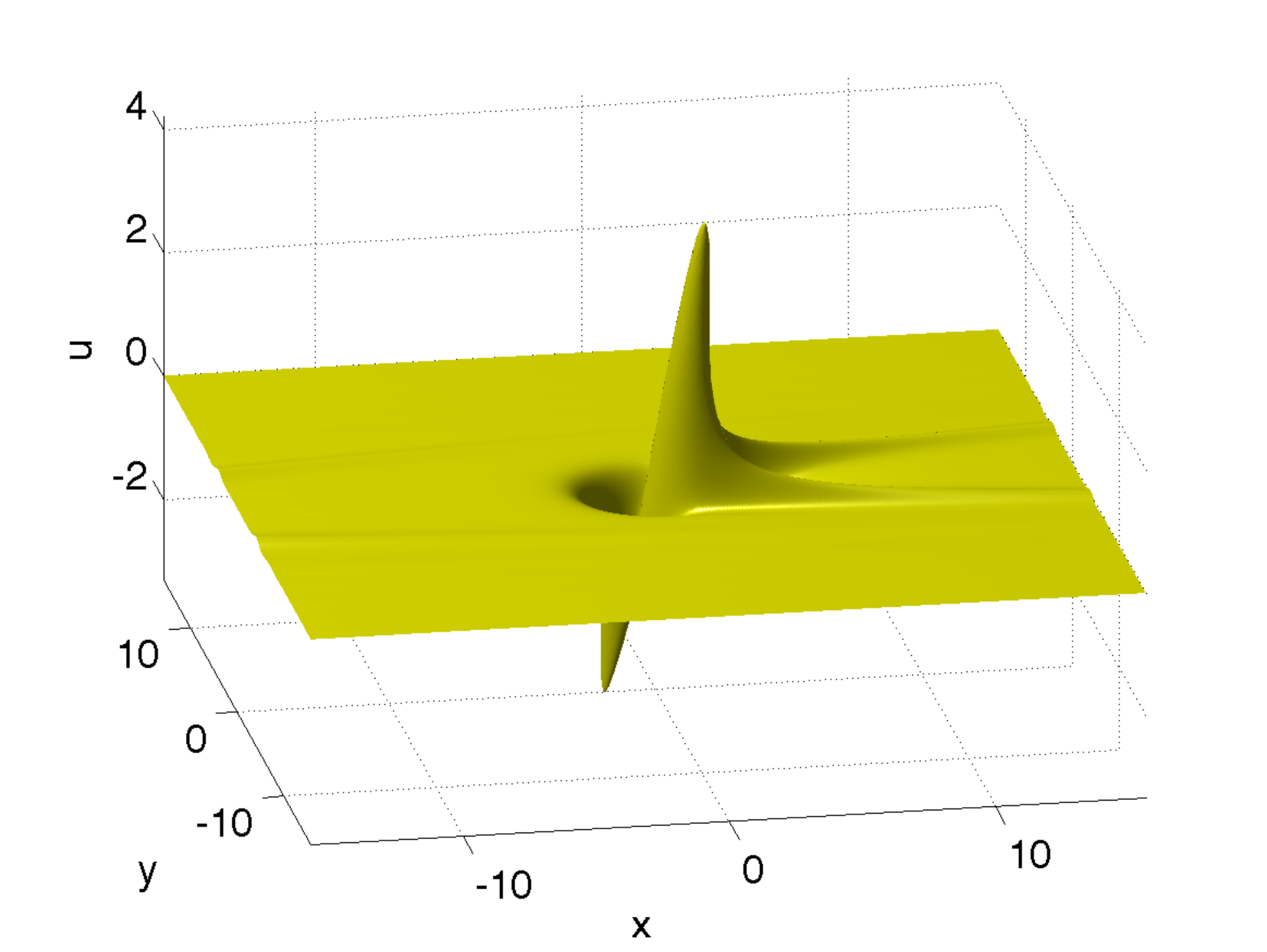}
  
    \vskip 0.1cm
\subfigure
{   \includegraphics[width=0.49\textwidth]{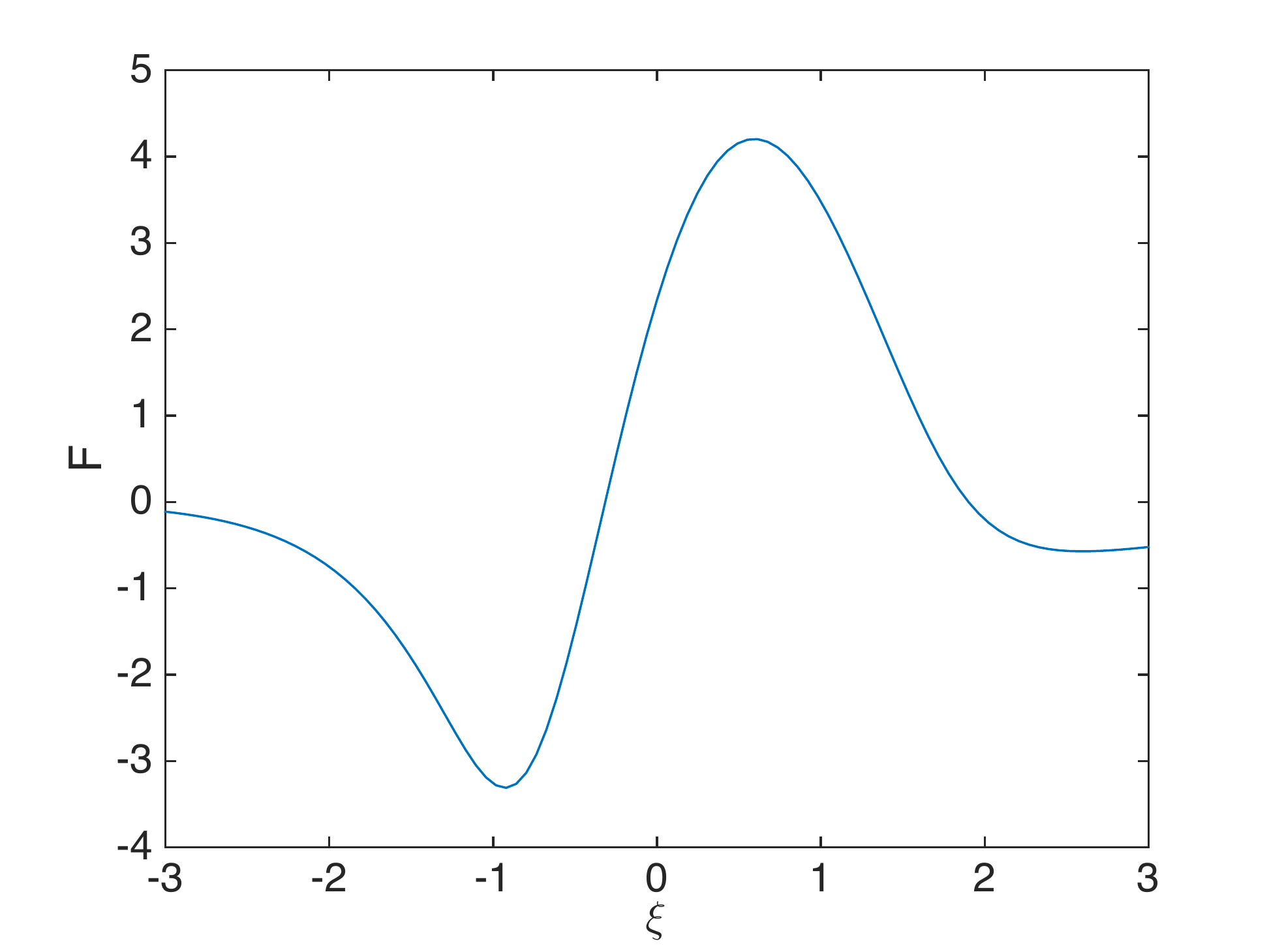} 
\includegraphics[width=0.49\textwidth]{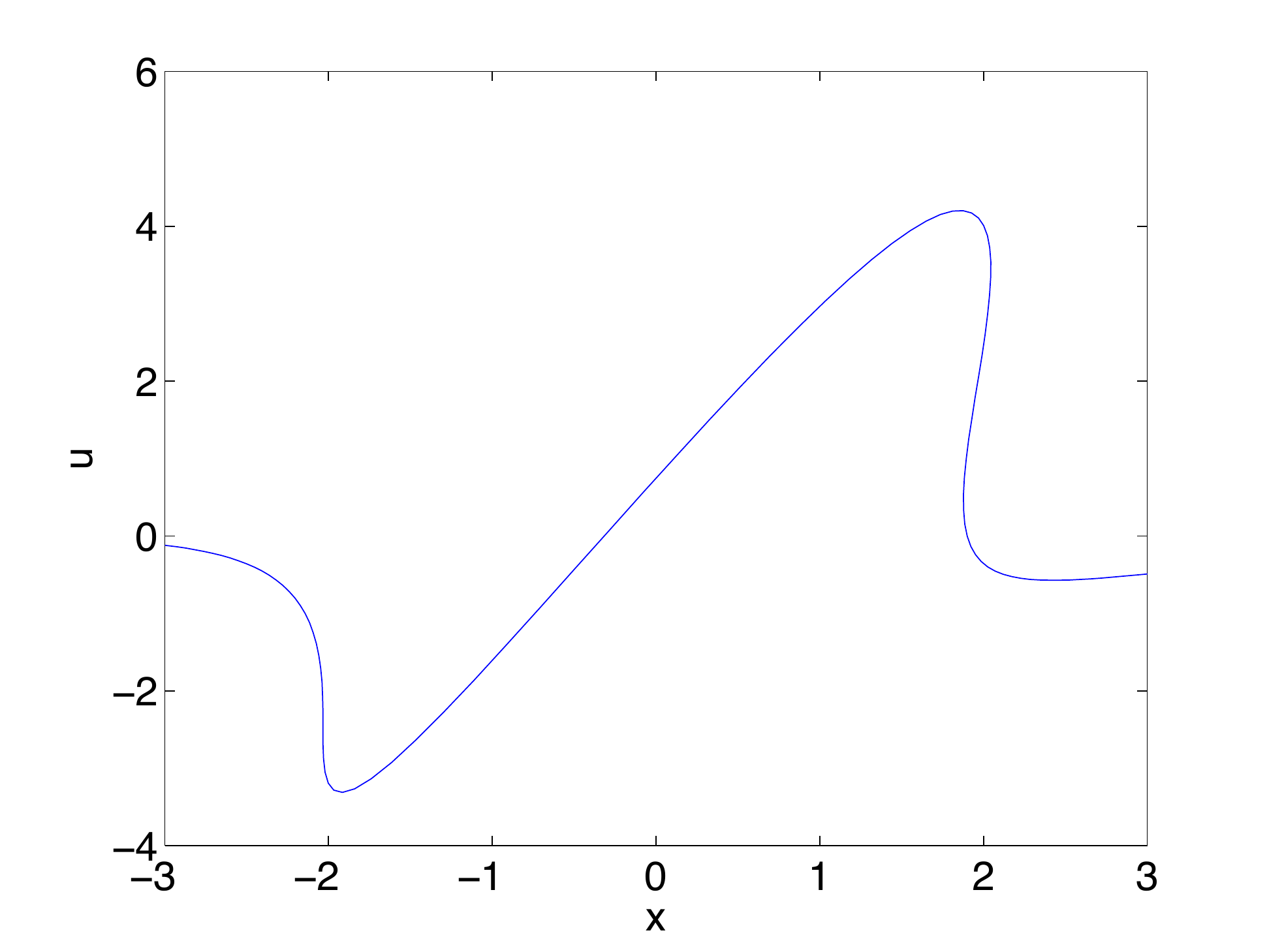}
  }
 \caption{Profiles obtained from a solution of the transformed equation
(\ref{eqF}) at $t = 0.300$, time of the second wave breaking event. On the 
left, the original solution $F(\xi,y,t)$ for initial data  (\ref{u0sym}); 
on the right, the profile $u(x,y,t)$ obtained using the transformation 
(\ref{implicit0}). The slices along the plane $y=0$ (bottom) make it clear 
that the profile $u(x,y,t)$ has overturned near $x=2$ (first breaking), 
and is at the point of breaking near $x=-2$ (second breaking). The profile 
of $F(\xi,y,t)$ remains smooth
and single valued. }
 \label{Fsechtc2}
\end{figure}
The corresponding profile $u(x,y,t)$ can be seen in Fig.~\ref{Fsechtc2} 
on the left. It is obtained by plotting $F(\xi,y,t)$ (shown on the right) 
as a function of $x=\xi+tF(\xi,y,t)$, as required by (\ref{implicit0}). 
For $t>t_c$  in a  neighborhood  of the blow-up point, one has that 
$x=\xi+tF(\xi,y,t)$ is not  invertible as a function of $\xi(x,y,t)$. 
However we can still perform a parametric plot of $u(x,y,t)$, which becomes 
a multivalued function  in the region near the first critical point 
$(x_{c},0,t_{c})$. This is even clearer from the cut along the $y=0$-axis 
shown on the bottom (recall that the critical points are all on the 
$x$-axis since the initial data are symmetric with respect to $y\to-y$, 
and since the dKP equation preserves this symmetry). Thus as for the 
solution to the Hopf equation via the characteristic method, a nonphysical 
solution which has overturned is obtained in the shock region. It is clear 
from the corresponding cut through $F(\xi,y,t)$ shown on the bottom left 
that $F$ remains smooth and single valued. 
\begin{figure}
\centering
\hskip 1cm
   \includegraphics[width=0.75\textwidth]{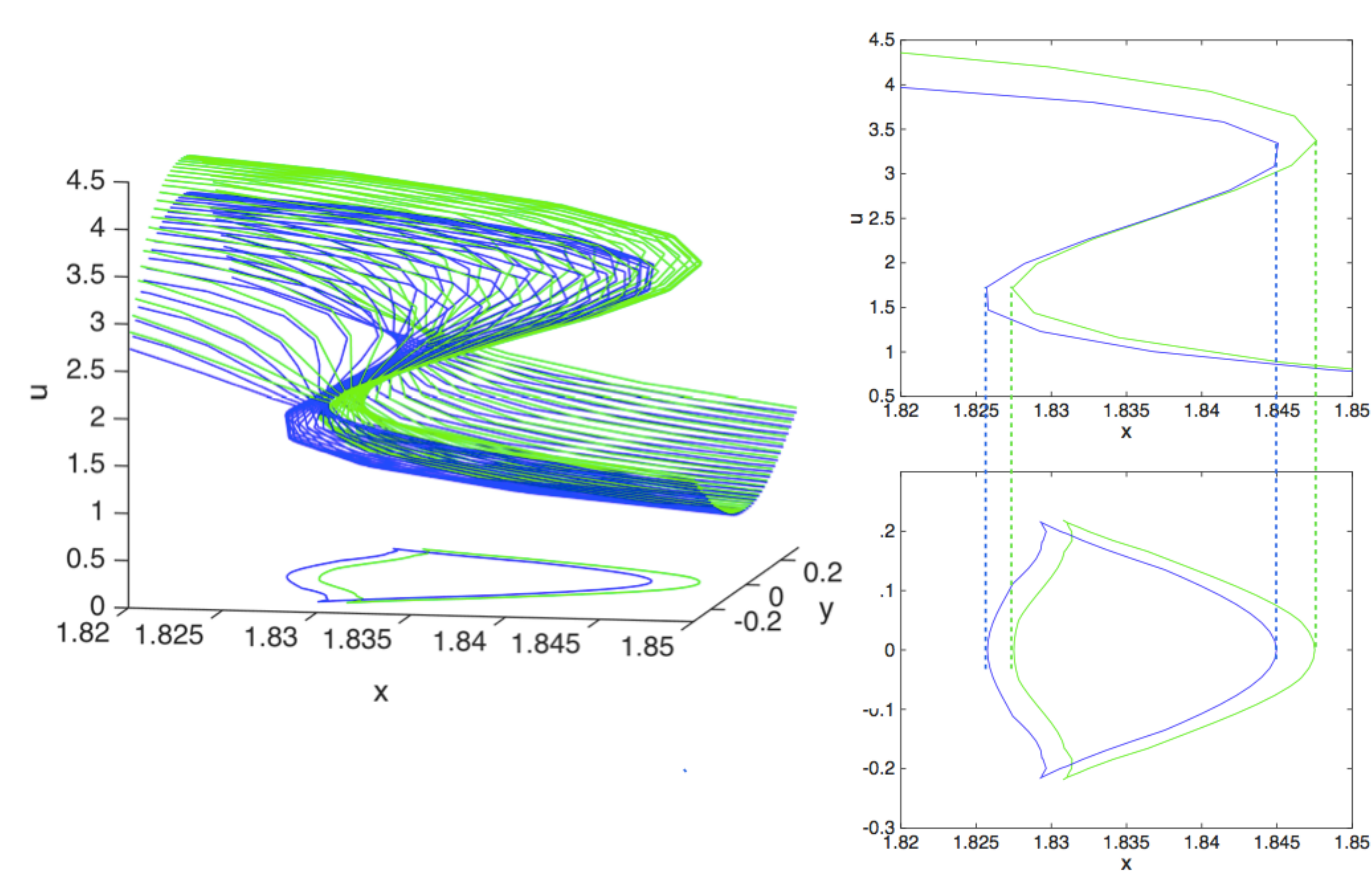}
\caption{On the left, the solution $u(x,y,t)$  (blue line) of 
the dKP equation  and its approximation (green line) (\ref{s_u}) for 
$t=0.24>t_c\approx0.222$. The  regions of multivaluedness of the solutions are 
projected on the $(x,y)$-plane. On the right top: a cut through $u(x,y,t)$
along $y=0$. On the right bottom: the corresponding multivalued 
regions  of $u(x,y)$ in the $(x,y)$-plane (blue line: numerical solution;
green line: local approximation.)
                  }
 \label{dKPcomparison}
\end{figure}

We can now test to which extent the asymptotic description of the 
overturned region in Section~\ref{sec:overturning}, which only becomes 
exact in the limit $t\sim t_{c}$, can approximate our numerical results. 
Recall that the profile is described by (\ref{s_u}), while the shape
of the overturned region is given by (\ref{X1Y1}),(\ref{slitX1Y1}).
In Fig.~\ref{dKPcomparison} we show a comparison between a numerical 
solution of the dKP equation, obtained through the transformation 
(\ref{implicit0}) (blue), with the local approximate  solution (\ref{s_u}) 
shown in green. At $t=0.24$, i.e. shortly after overturning at 
$t_c=0.222$, there is good agreement in the description of the 
multivalued region. On the left, $u(x,yt)$ is shown in a perspective
plot, on the top right an s-curve is produced by a cut along the 
$y=0$ plane. If corresponding cuts are considered for each value of $y$, 
a lip-shaped region is obtained inside which the profile has overturned
(bottom right). 

\begin{figure}
  \centering
 \includegraphics[width=0.6\textwidth]{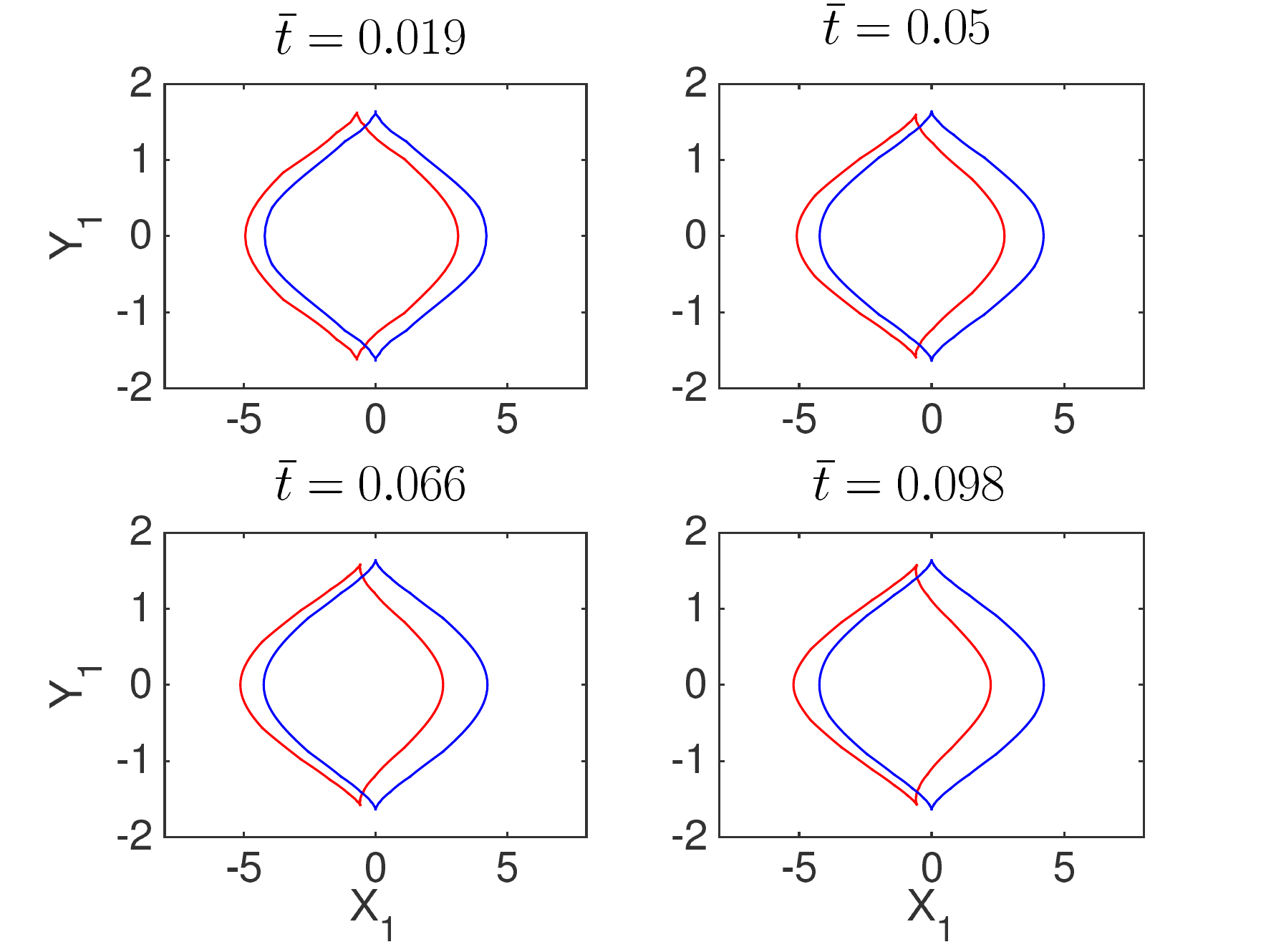}
  \caption{Multivalued region of the solution of the dKP equation 
as found from $\Delta(\xi,y,t)=0$  for the initial data (\ref{u0sym}). 
Results are written in selfsimilar rescaled coordinates 
$X_{1}$ and $Y_{1}$  defined by (\ref{X1Y1}) for several values of $\bar{t}$ 
(red lines). The corresponding asymptotic boundary (\ref{slitX1Y1}), 
shown in blue, is time-independent by construction. }
 \label{Fsechcontour4}
\end{figure}

To test for the self-similar properties of the multivalued region, 
in Fig.~\ref{Fsechcontour4} we show the numerical result as function of 
the rescaled coordinates $X_1$, $Y_{1}$, which are defined by (\ref{X1Y1}) 
(red lines). Good agreement is seen with the asymptotic prediction 
(\ref{slitX1Y1}) (blue lines), in particular for small values of three time 
distance $\bar{t}$ from the gradient catastrophe, as expected. The fact
that the numerical results stay time independent to a good approximation 
demonstrates that the typical scales of the solution agree with the 
prediction (\ref{X1Y1}): the width of the region scales like $\bar{t}^{3/2}$,
its height like $\bar{t}^{1/2}$. 

\begin{figure}
    \includegraphics[width=0.49\textwidth]{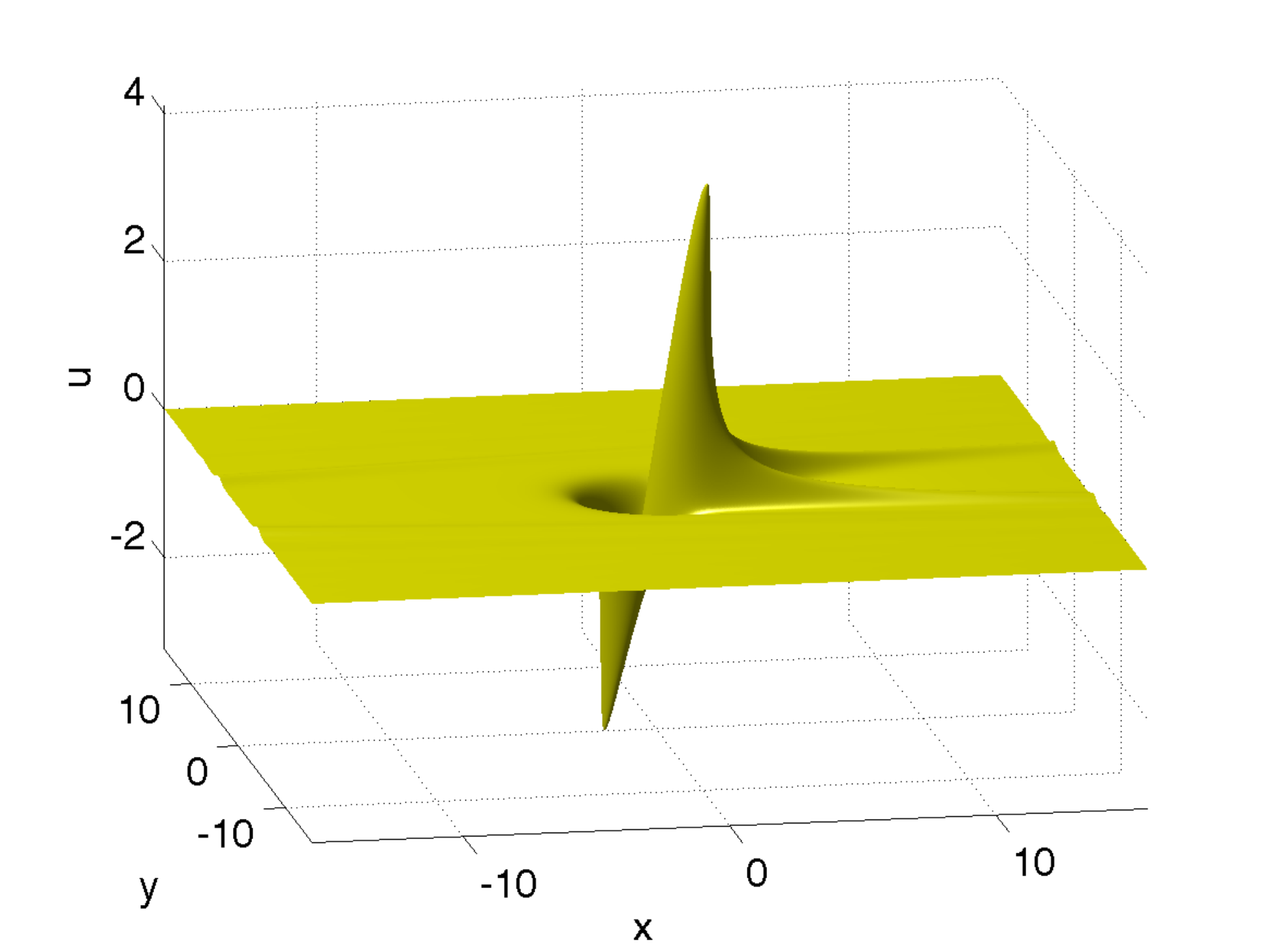}
    \includegraphics[width=0.49\textwidth]{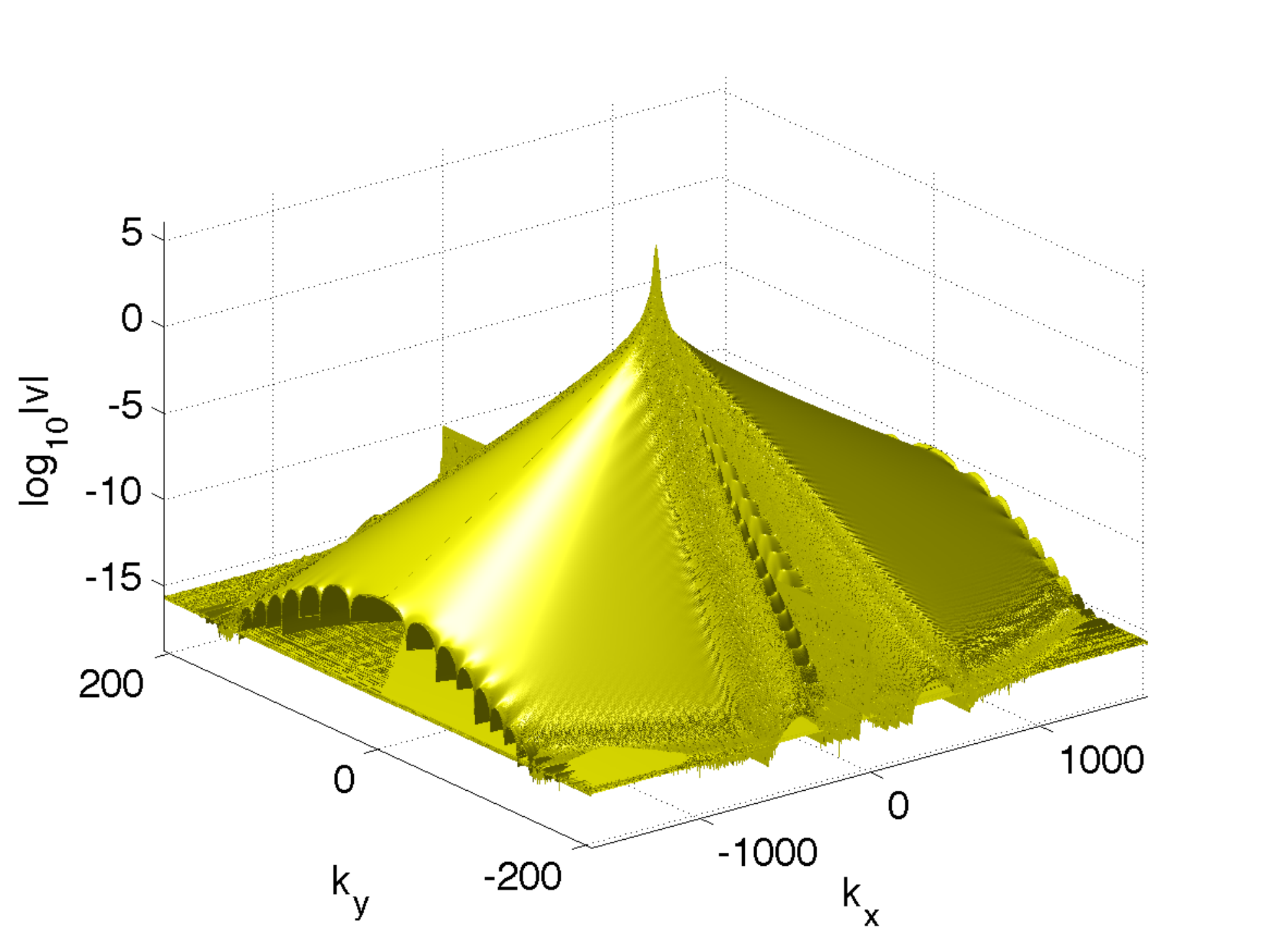}
 \caption{Numerical solution to the  dissipative dKP equation (\ref{dKPdis}) 
with $c=0$ and $\epsilon=0.01$ for initial data (\ref{u0sym}) at time
$t=0.32$ on the left, and the corresponding Fourier coefficients on 
 the right.}
 \label{dKPdissech}
\end{figure}
We now turn to the numerical solution of the dissipative dKP equation 
(\ref{DKP}), and to the comparison with our asymptotic theory, 
which is given by (\ref{KPPearcy}) in the general case, and by 
(\ref{KPPeracy}) for symmetric initial data. To resolve the strong 
gradients in the solutions to the dissipative dKP equation (\ref{dKPdis})
that occur for small $\epsilon$, much higher resolution is needed than 
for the solution of (\ref{Fev}) for the same initial data. For 
$\epsilon=0.01$ (with $c=0$) we use $N_{x}=2^{14}$, $N_{y}=2^{10}$ and 
$N_{t}=5000$ to find the solution of (\ref{dKPdis}) with initial data 
(\ref{u0sym}) at $t=0.32$, shown in Fig.~\ref{dKPdissech} on the left. 
At this value of $\epsilon$, the total loss of the $L^2$ norm 
(cf. (\ref{diss})) is of the order of $2\%$. 
A comparison between the dKP solution and 
the Fourier coefficients,  shown on the right, decay to below $10^{-10}$, 
as for the solutions to (\ref{Fev}). To achieve higher resolutions, 
parallel computation would be needed. 

\begin{figure}
\centering
\hspace{3cm}
 \includegraphics[width=0.95\textwidth]{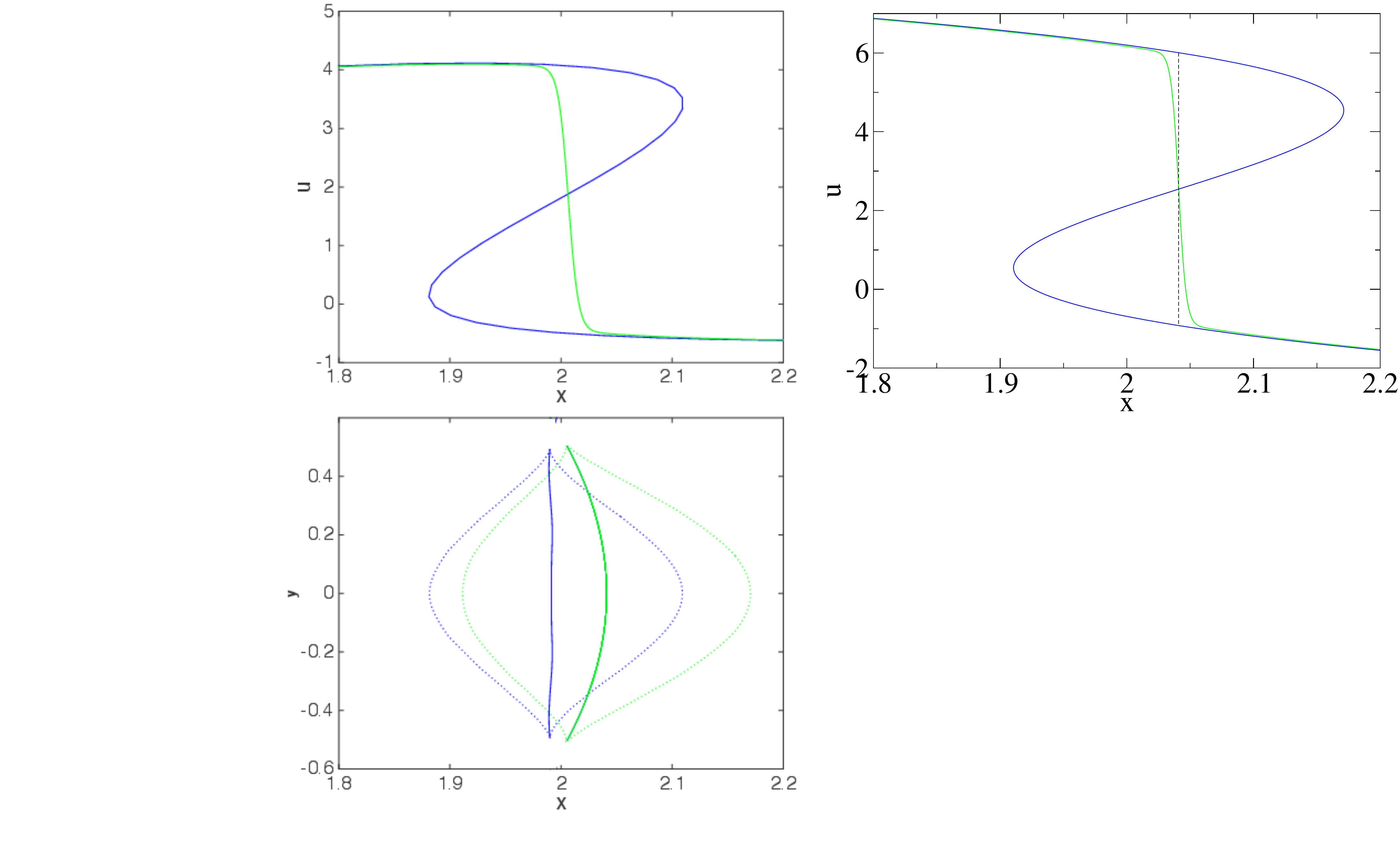}
 \caption{Top left: numerical solutions to the dKP equation (blue) and 
to the dissipative dKP equation (\ref{dKPdis}) (green), for
$c=0$ and $\epsilon=0.01$, using symmetric  initial data (\ref{u0sym}).
Shown is a slice along the line $y=0$ at $t=0.32>t_c=0.222$. 
Top right: the asymptotic approximations (\ref{u_exp}) and (\ref{KPPearcy}) 
to the same solutions; the dashed line marks the shock position $X=0$.
Bottom: the dotted lines mark the multivalued
regions for $t=0.32$, according to the numerical solution to the 
dKP equation (blue), and according to the asymptotic theory 
(\ref{lip_symm}) (green). The green solid line is the asymptotic 
prediction for the shock front, as given by (\ref{front_s}), and 
the blue solid line is a numerical estimate based on the inflection 
point of the dKP solution.}
 \label{dKPdissech2}
\end{figure}

In Fig.~\ref{dKPdissech2} (top left), we show a slice through the 
same dissipative solution at $y=0$ (green line), together with 
the corresponding dKP solution, which has become multivalued, as 
$\bar{t} \approx 0.1$. The dissipative solution exhibits a sharp 
front close to where the shock discontinuity is expected to be. 
Both curves are to be compared to our asymptotic 
results, shown on the top right, with the s-curve (\ref{s_u}) shown in 
blue, and the dissipative asymptotics (\ref{KPPeracy}) in green. 
The sharp front is seen to be localized around the theoretical shock
position, shown as the vertical dashed line. 
Since $\bar{t}$ is only moderately small, there exists a 30\% difference 
in the height of the s-curve, but otherwise the overturning of the dKP 
equation is well reproduced. Within these limitations, the shape and 
width of the shock front, as well as the front position within the 
s-curve, are very well reproduced. 

In the bottom graph of Fig.~\ref{dKPdissech2}, we report the multivalued
regions, as well as the position of the shock front, as given by the
numerical solution (blue curves, with the shock front as the solid line), 
and our asymptotic theory (green curves, shock front solid). Once more,
there is fair agreement in the shape and size of the lip-shaped multivalued
regions (dashed lines), described by the dKP equation. The numerical shock
position is estimated from the inflection point of the dKP solution, 
the theoretical prediction is the curve $X=0$.

In Fig.~\ref{fig:pearceysymc}, we show the solution to the dissipative 
dKP equation (\ref{DKP}) for $\epsilon=0.01$ and the asymptotic description 
(\ref{KPPearcy}) for the symmetric initial data (\ref{u0sym}) at the 
critical time in the vicinity of the critical point. While the asymptotic 
formula provides the best local approximation being best near the 
critical point, it can be seen to also correctly reproduce the 
$y$-dependence. 
\begin{figure}
\centering
   \includegraphics[width=0.49\textwidth]{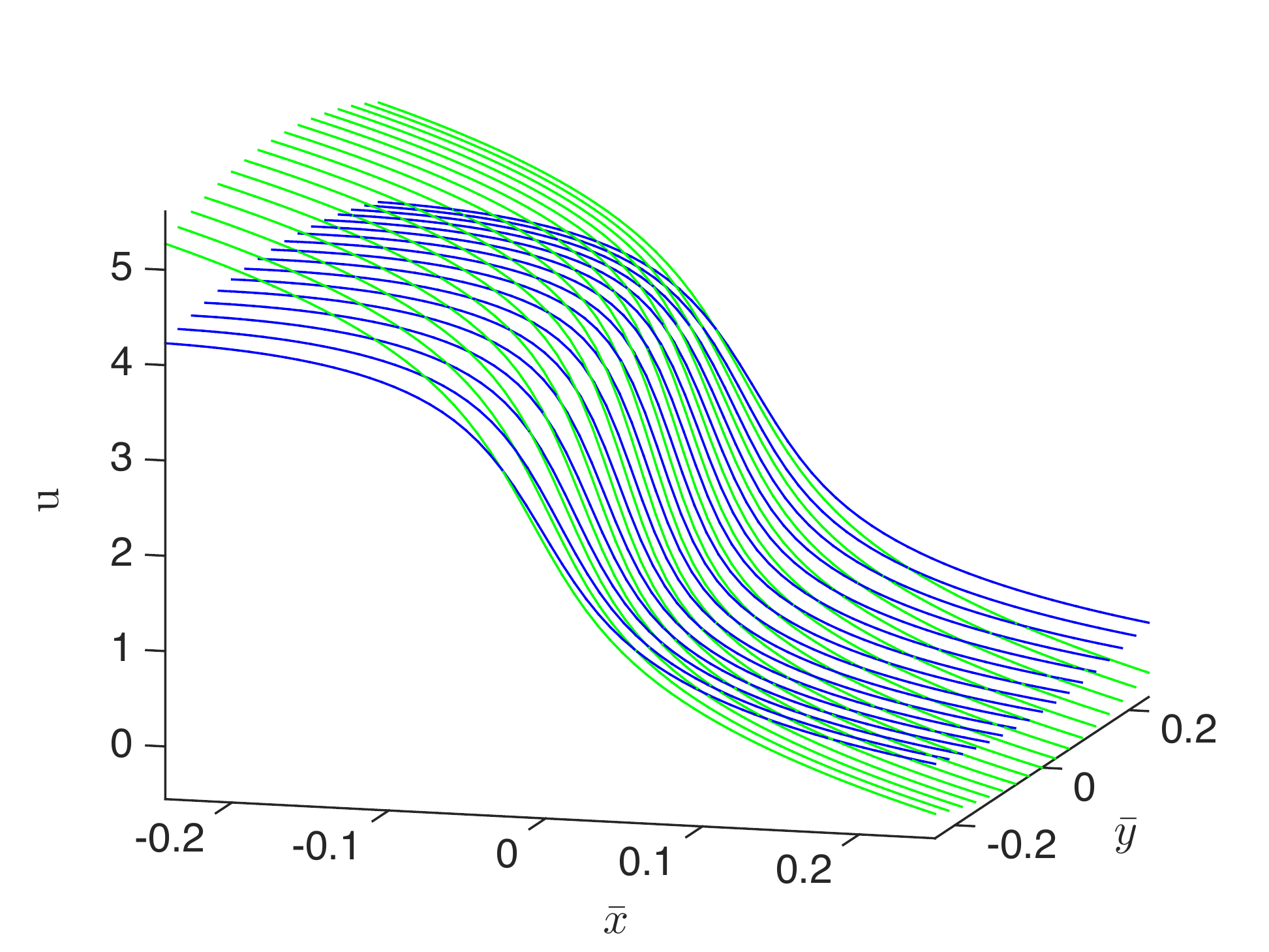} 
   \includegraphics[width=0.49\textwidth]{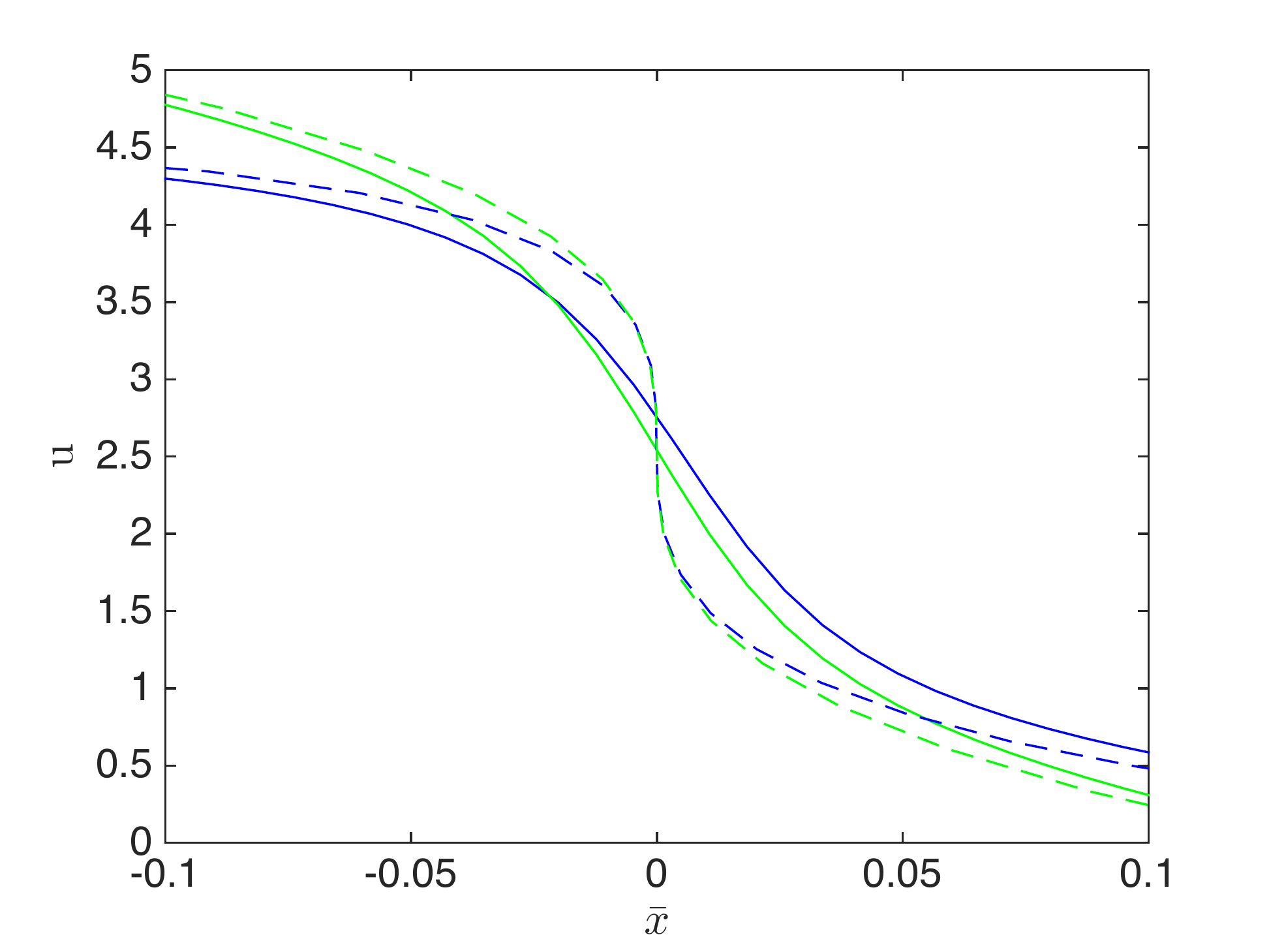}
  \caption{On the left, in blue the  solution to the dissipative dKP 
 equation (\ref{DKP}) for $\epsilon=0.01$ and 
  the symmetric initial data (\ref{u0sym}) at the critical time 
  $t_{c}=0.222$   and near the critical point, and in green  the 
  asymptotic solution (\ref{KPPearcy})  given by the Pearcey 
  integral. On the right the same plot along the line $y=0$. The dashed  
  blue line is the solution of dKP equation  and the green dashed line is the solution of the approximation (\ref{u_exp}) to the dKP solution. }
 \label{fig:pearceysymc}
   \end{figure}

The approximation is also valid for small, nonzero values of 
$\bar{t}$ as can be seen in Fig.~\ref{fig:pearceysym9t6} where the 
same situation as in Fig.~\ref{fig:pearceysymc} is shown on the 
slice $y=0$  for several values of $\bar{t}$. 
\begin{figure}
\centering
   \includegraphics[width=0.7\textwidth]{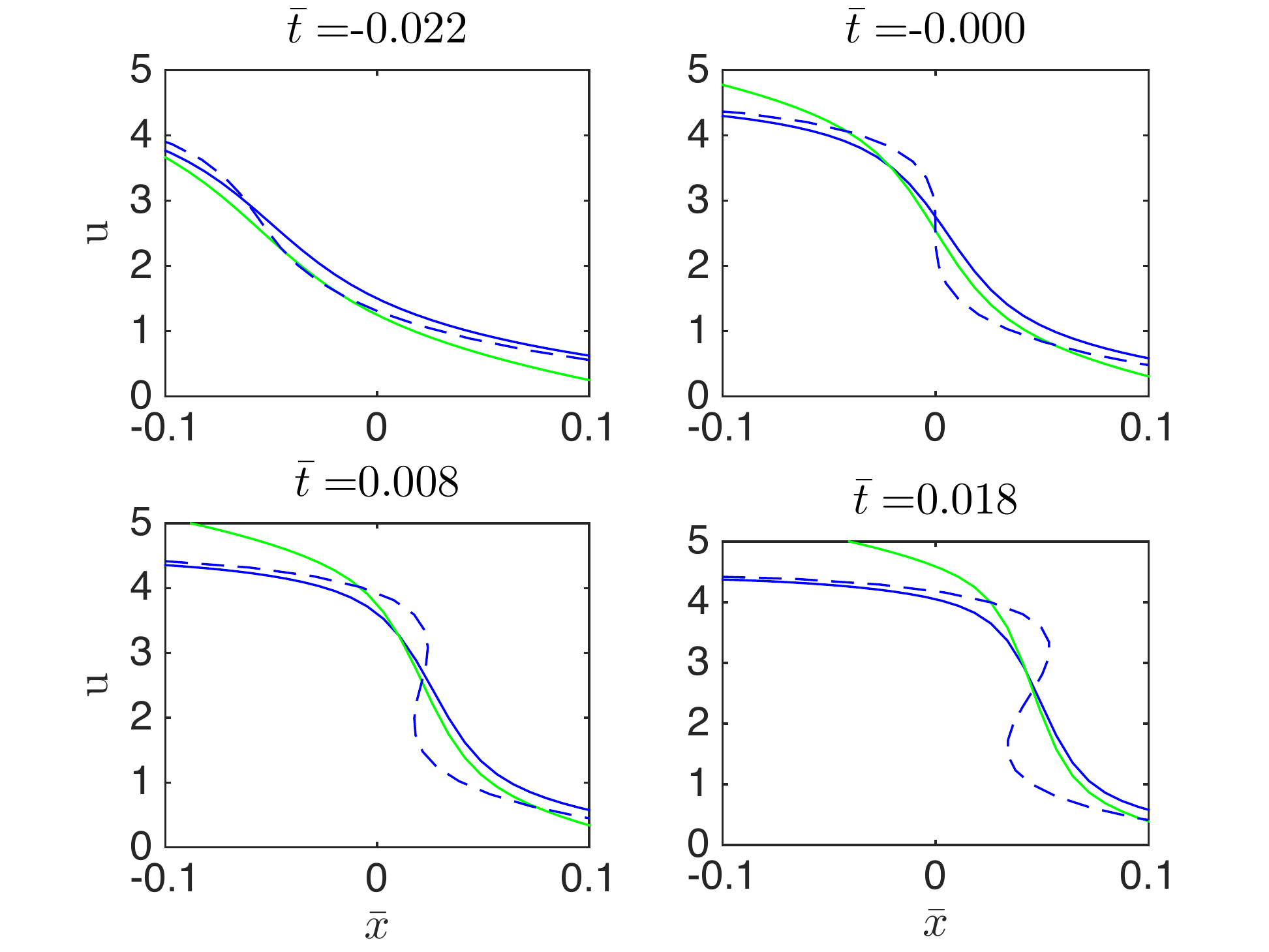} 
  \caption{Solution to the dissipative dKP equation (\ref{DKP}) for 
  $\epsilon=0.01$ and 
  the symmetric initial data (\ref{u0sym}) in blue, the 
 Pearcey  asymptotic solution (\ref{KPPearcy}) in green and the (weak) dKP 
  solution dashed on the line $y=0$ for several values of $\bar{t}$. 
}
 \label{fig:pearceysym9t6}
   \end{figure}

\subsection{Nonsymmetric initial data}
In this section we consider two different initial profiles which 
are not symmetric with respect to $y\to-y$. The first, 
\begin{equation}
    u(x,y,0) = 6\partial_{x}\left\{(x+1)(y-1)e^{-x^{2}-y^{2}}\right\},
    \label{u1}
\end{equation}
still retains a radial symmetry for $x^{2}+y^{2}\to\infty$. As 
seen in Table~\ref{table2}, we can follow the evolution through 
two successive gradient catastrophes. The second profile, 
\begin{equation}
    u(x,y,0)=6\partial_xe^{-x^2-5y^2-3xy},
    \label{u2}
\end{equation}
does not possess radial symmetry for large $x^{2}+y^{2}$, and we are 
able to compute the first gradient catastrophe only, whose critical 
parameters are also given in Table~\ref{table2}.
\begin{table}[h]
 \centering
    \leavevmode
\begin{center}\begin{tabular}{|c| c | c |c|c|c|c|}
\hline
Breaking events &Initial data&$t_c$&$x_c$&$y_c$&$u_c$&$\xi_{c}$\\
\hline First breaking &
$ 6\partial_{x}\left\{(x+1)(y-1)e^{-x^{2}-y^{2}}\right\}
$&0.0832&-1.210& -0.368&-4.958&-0.798\\
\hline
Second breaking &
$ 6\partial_{x}\left\{(x+1)(y-1)e^{-x^{2}-y^{2}}\right\}
$&0.1070&2.004&-0.368&4.4066&1.534\\
\hline
First Breaking&$6\partial_x(e^{-x^2-5y^2-3xy})$&0.086&0.088& 
-0.245&-1.477&0.215\\
\hline
\end{tabular}
\label{table2}
\caption{Critical parameters for the first two wave breaking events, with
weakly asymmetric initial data (\ref{u1}). For the strongly asymmetric 
initial data (\ref{u2}) only the first breaking could be computed.}
\hskip 100pt
\end{center}
\end{table}

\begin{figure}
\hspace{60pt}
    \includegraphics[width=0.4\textwidth]{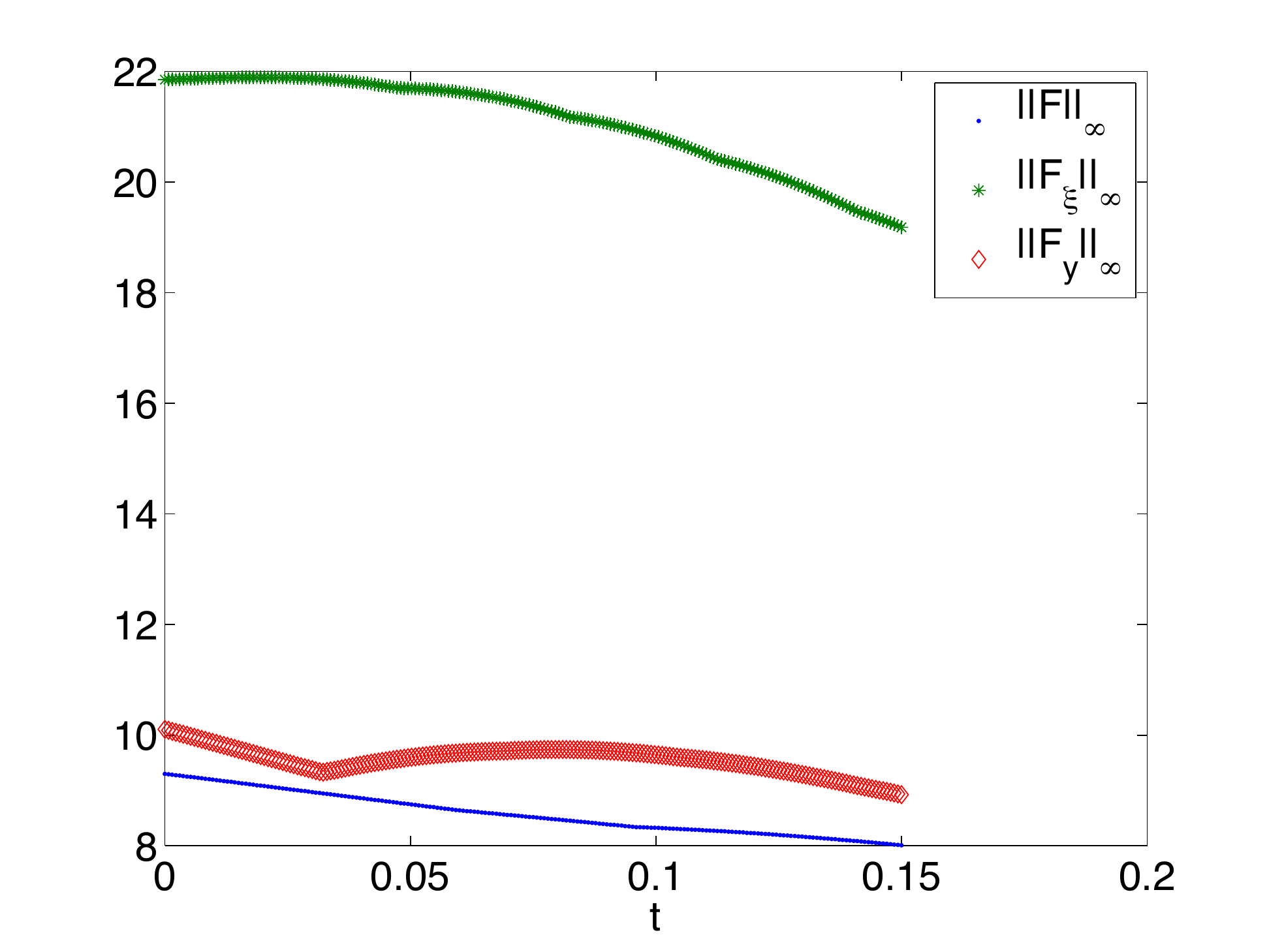}
    \includegraphics[width=0.4\textwidth]{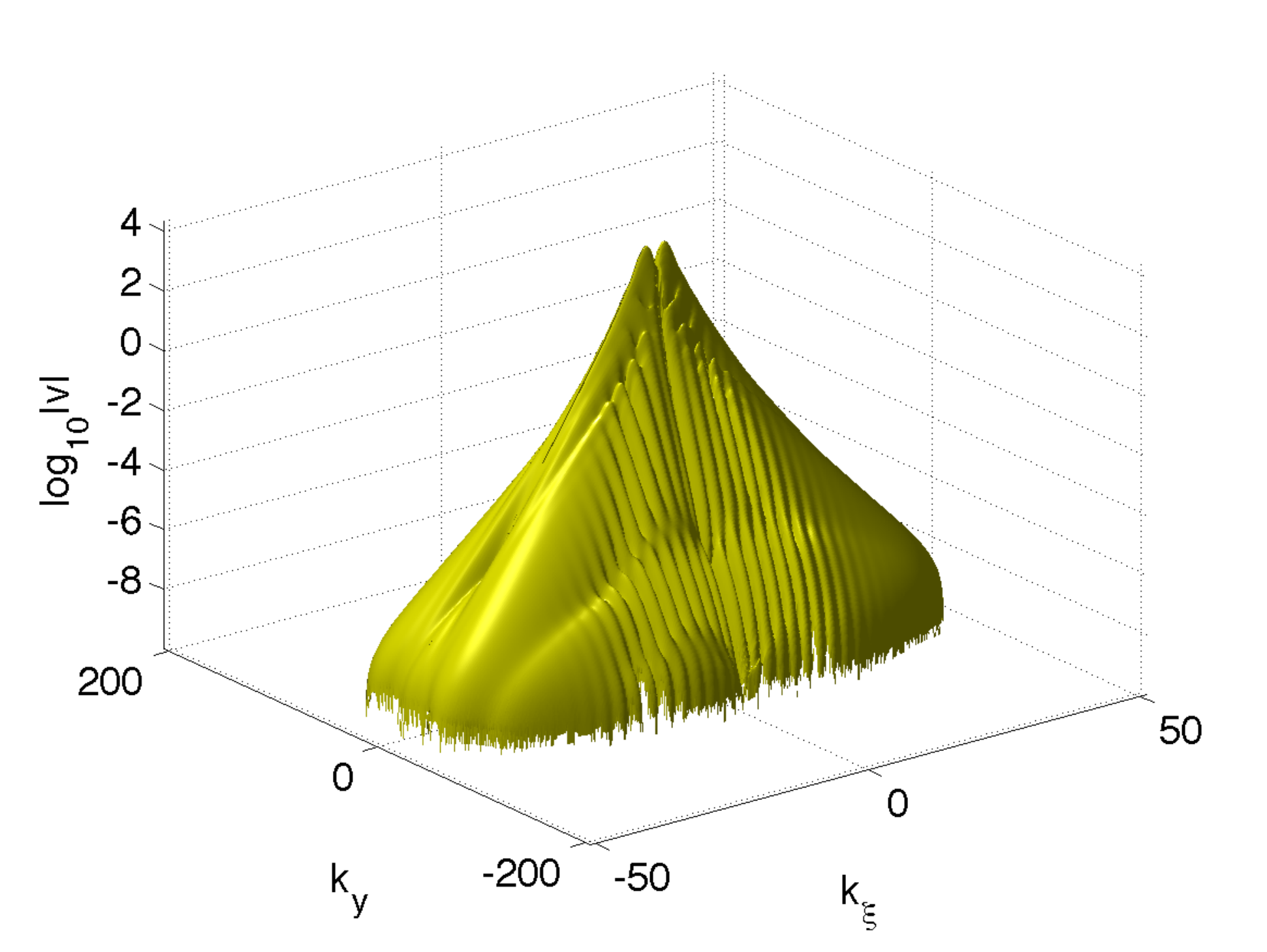}
 \caption{Same as Fig.~\ref{Fsechnorm}, but with initial data (\ref{u1})
(left). The Fourier coefficients on the right are shown for $t=0.15$. 
}
 \label{Fexpnorm}
\end{figure}
To solve the Cauchy problem with initial data (\ref{u1}) for the 
dKP equation (\ref{eqF}), we use $N_{x}=2^{9}$ and $N_{y}=2^{11}$ Fourier 
modes for $(x,y)\in[-5\pi,5\pi]^{2}$ and $N_{t}=5000$ time steps for 
$t\leq 0.15$. The first critical time is reached at $t_c=0.08323\ldots$,
the second critical time is $\tilde{t}_{c}=0.1070\ldots$; all other 
critical parameters are reported in Table~\ref{table2}. 
The relative computed $L^2$ norm is conserved to the order 
of $10^{-14}$, and the Fourier coefficients decrease to the order of 
the Krasny filter as can be seen in Fig.~\ref{Fexpnorm} (left). 
As seen in the same figure on the left, the $L^{\infty}$ norm of the solution 
$F$ and the norm of its gradient also appear to decrease for large $t$,
so again there is no indication of a blow-up of the solution. However, 
to be able to run the code for longer times, larger computational domains 
would have to be used.

\begin{figure}
\hspace{20pt}
 \includegraphics[width=0.5\textwidth]{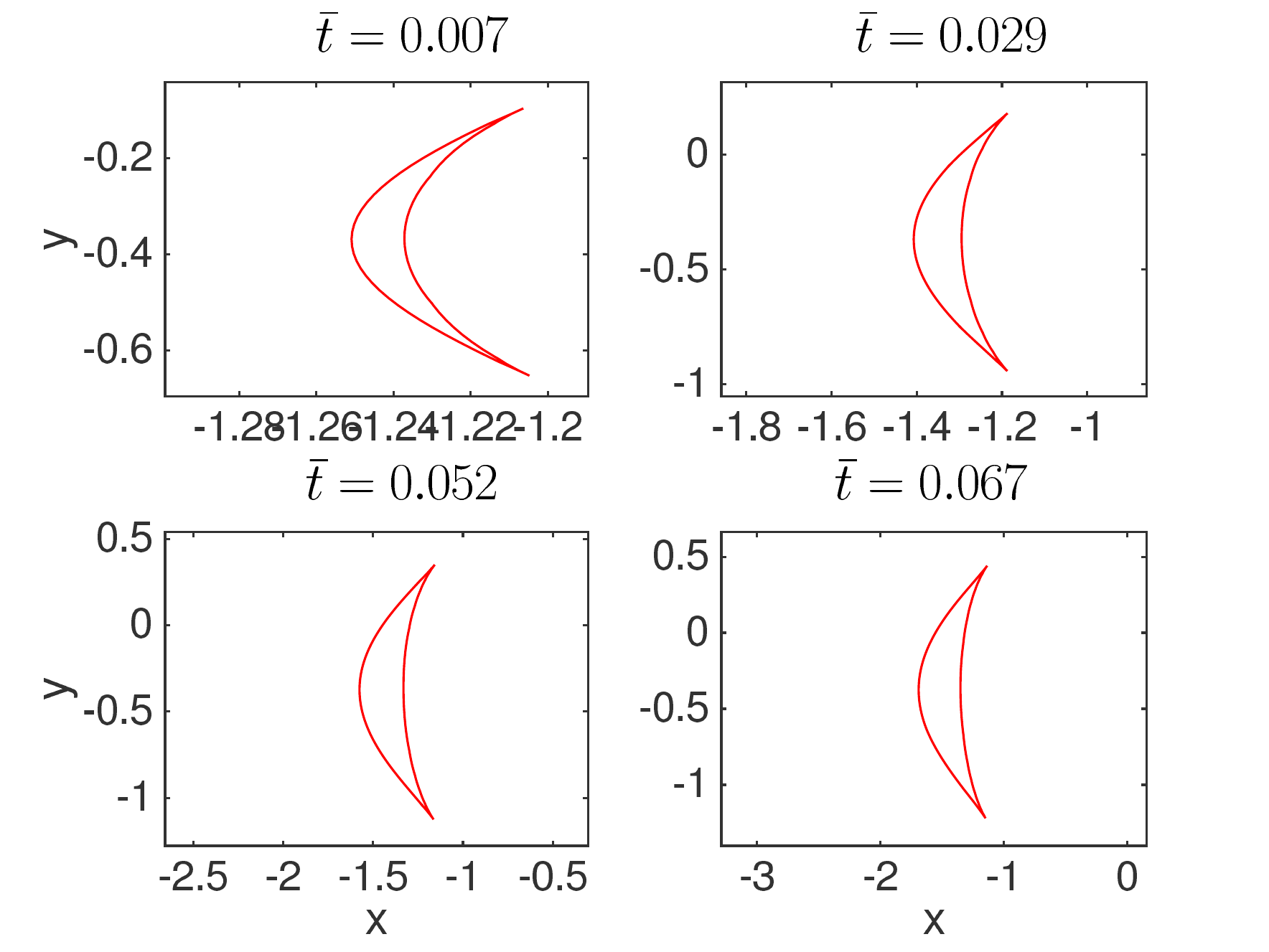}
    \includegraphics[width=0.5\textwidth]{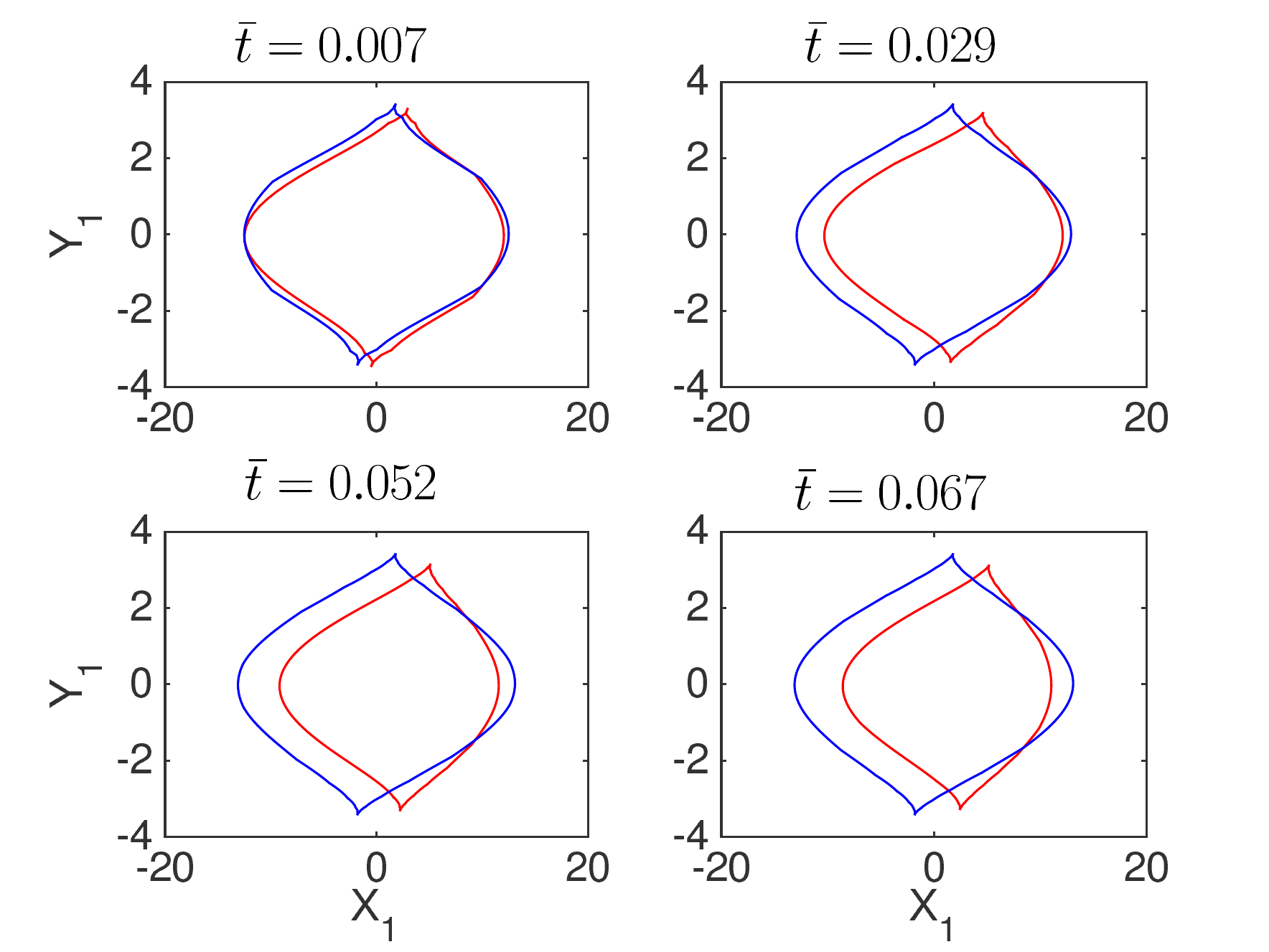}
 \caption{Left: boundary of the multivalued region found from a numerical
solution to the dKP equation for the initial data (\ref{u1}), 
for several values of $t>t_c=0.08323\ldots$  in the original (x,y) variables.
Right: The red boundaries on the right are the same data represented
in self-similar variables $X_1$ and $Y_1$ as defined in (\ref{X1Y1}), 
predicted to be time-independent by our asymptotic theory. 
The corresponding self-similar boundary, given by (\ref{slitX1Y1}),
is plotted in  blue.} 
\label{Fexp}
\end{figure}
On the left of Fig.~\ref{Fexp}, we trace the boundary of the multivalued 
regions  of $u(x,y,t)$  at four times shortly after the first gradient 
catastrophe; the times $\bar{t}$ relative to the singularity are reported 
on the top of each graph. On the right of the same figure, the same 
multivalued regions are plotted as functions of the rescaled coordinates 
$X_{1}$ and $Y_{1}$ defined in (\ref{X1Y1}). Once more, in the rescaled
coordinates the shape of the multivalued region is almost constant,
and agrees well with the theoretical prediction, shown in blue. Note 
the slight asymmetry of the lip shape with respect to the reflection 
symmetry $y\to-y$.

\begin{figure}
    \includegraphics[width=0.49\textwidth]{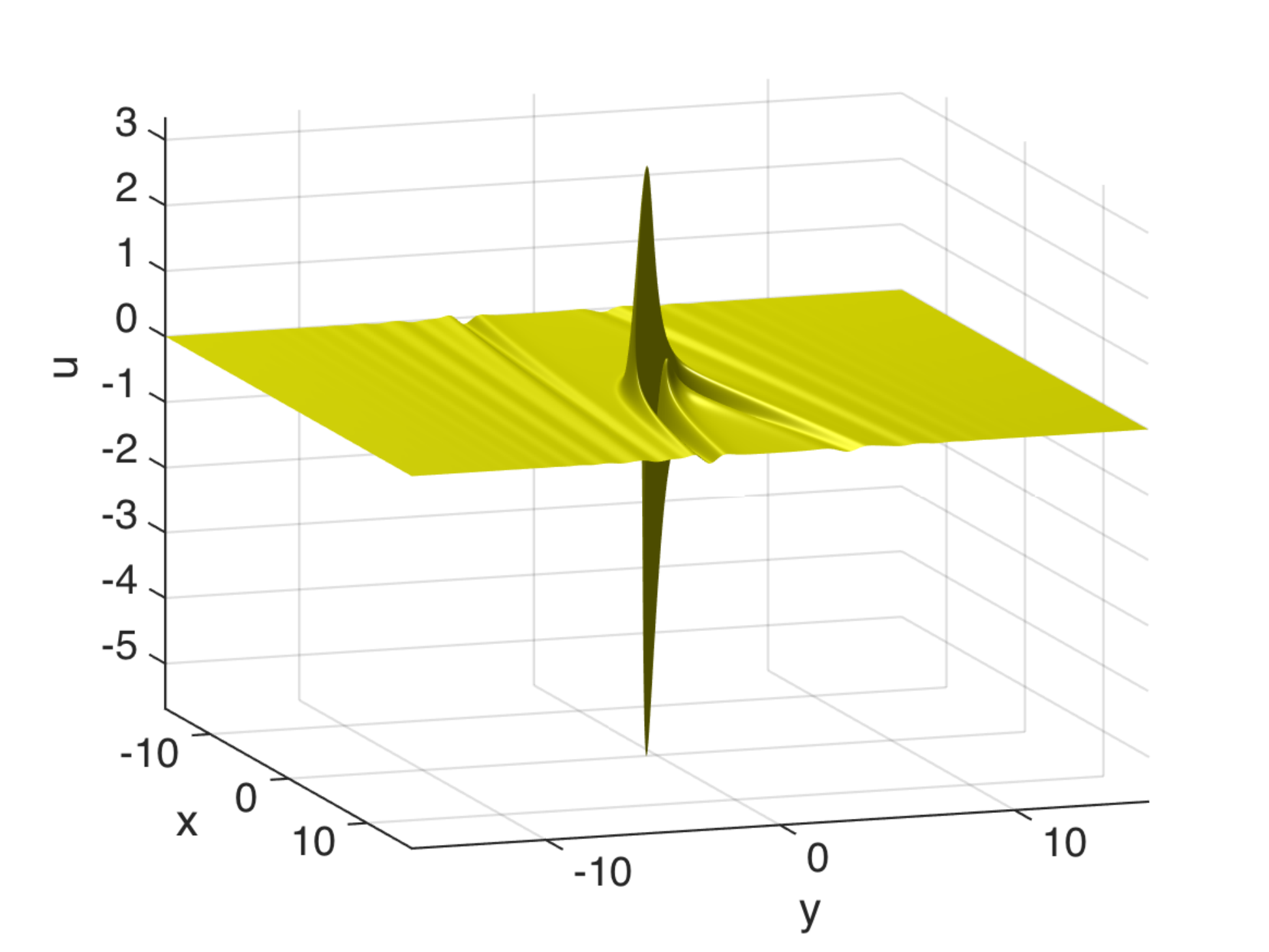}
    \includegraphics[width=0.49\textwidth]{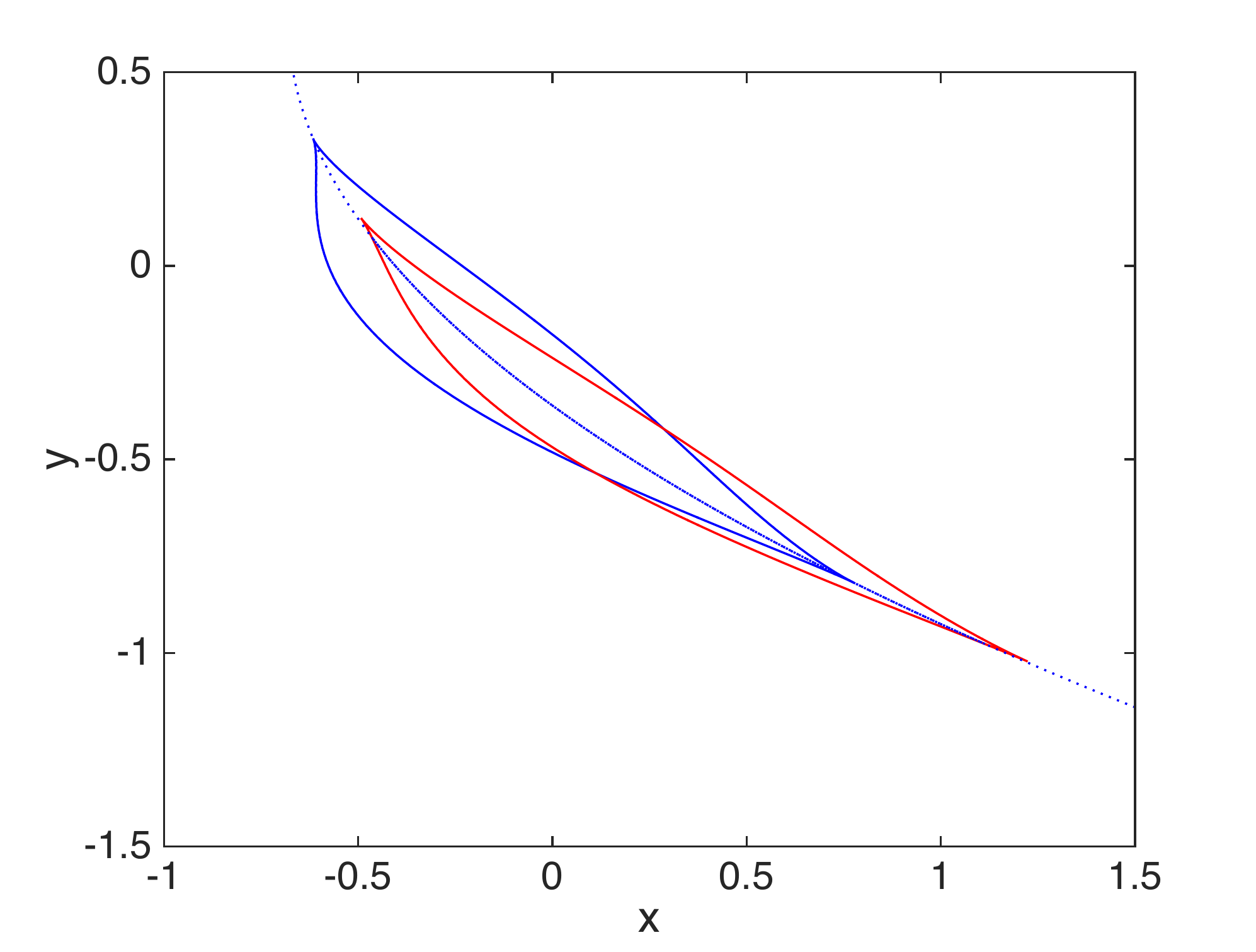}
 \caption{Left: numerical solution to the dKP  equation (\ref{dKP}) 
for strongly asymmetric initial data (\ref{u2}) at  $t=0.15>t_c=0.087$. 
Right: The corresponding  contour of the multivalued region 
$\Delta(\xi,y,t)=0$ (red), compared to the asymptotic theory 
(\ref{slitX1Y1}) (blue); the dashed line corresponds to $X=0$ as given 
by (\ref{front_s}).
       }
 \label{Fexp2}
\end{figure}
\begin{figure}
\subfigure{
    \includegraphics[width=0.35\textwidth]{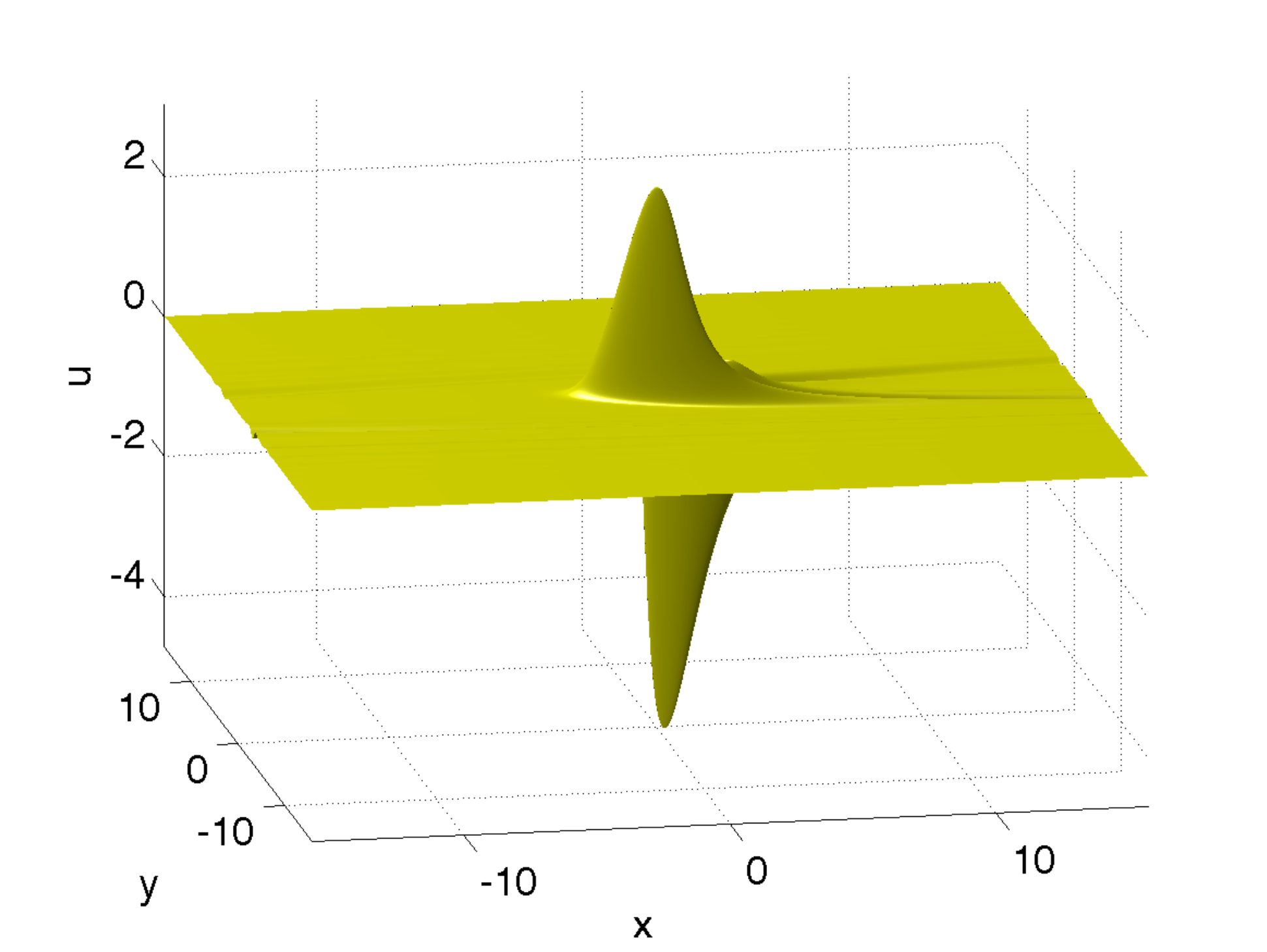}
     \includegraphics[width=0.35\textwidth]{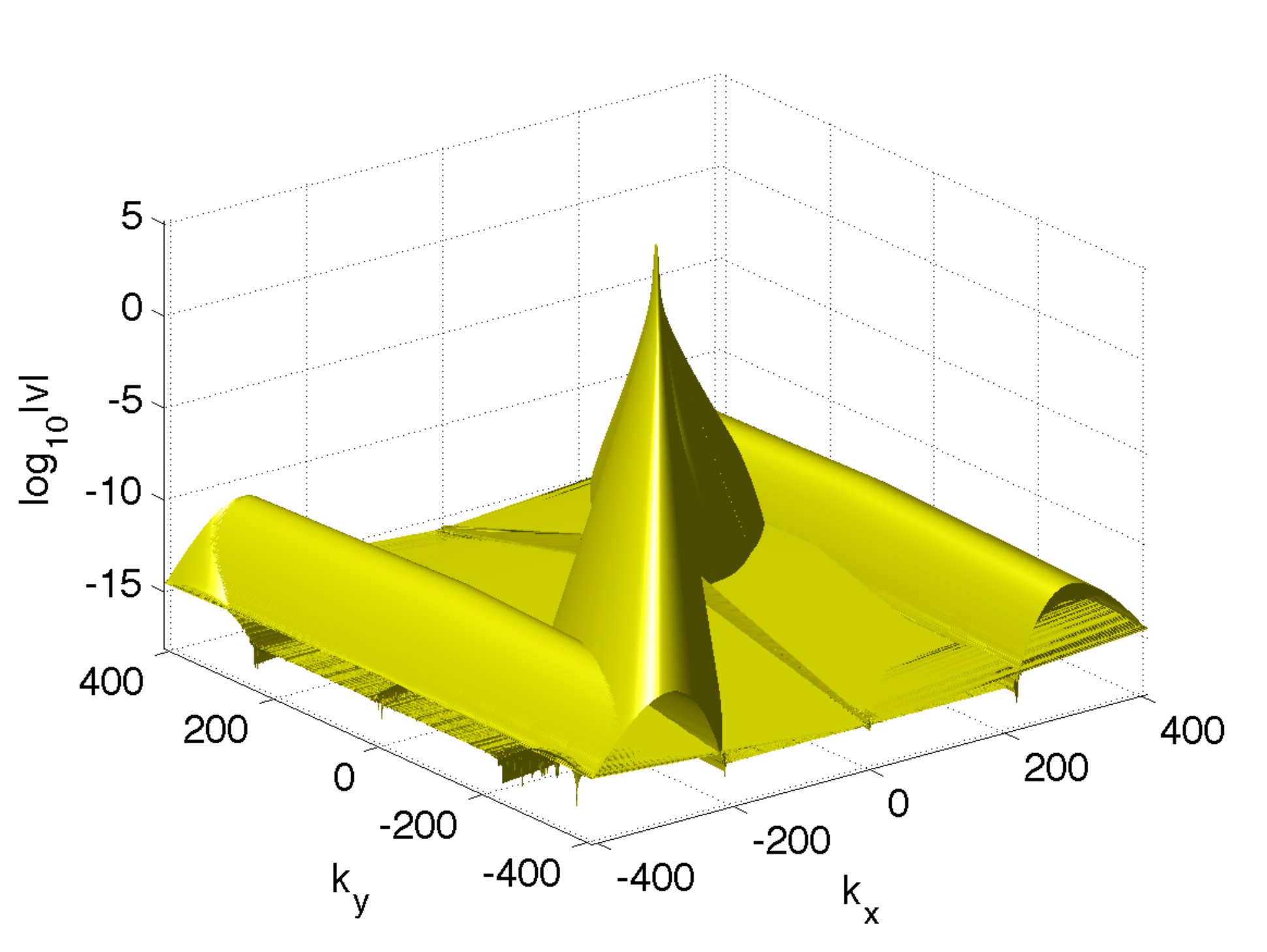}  
       \includegraphics[width=0.35\textwidth]{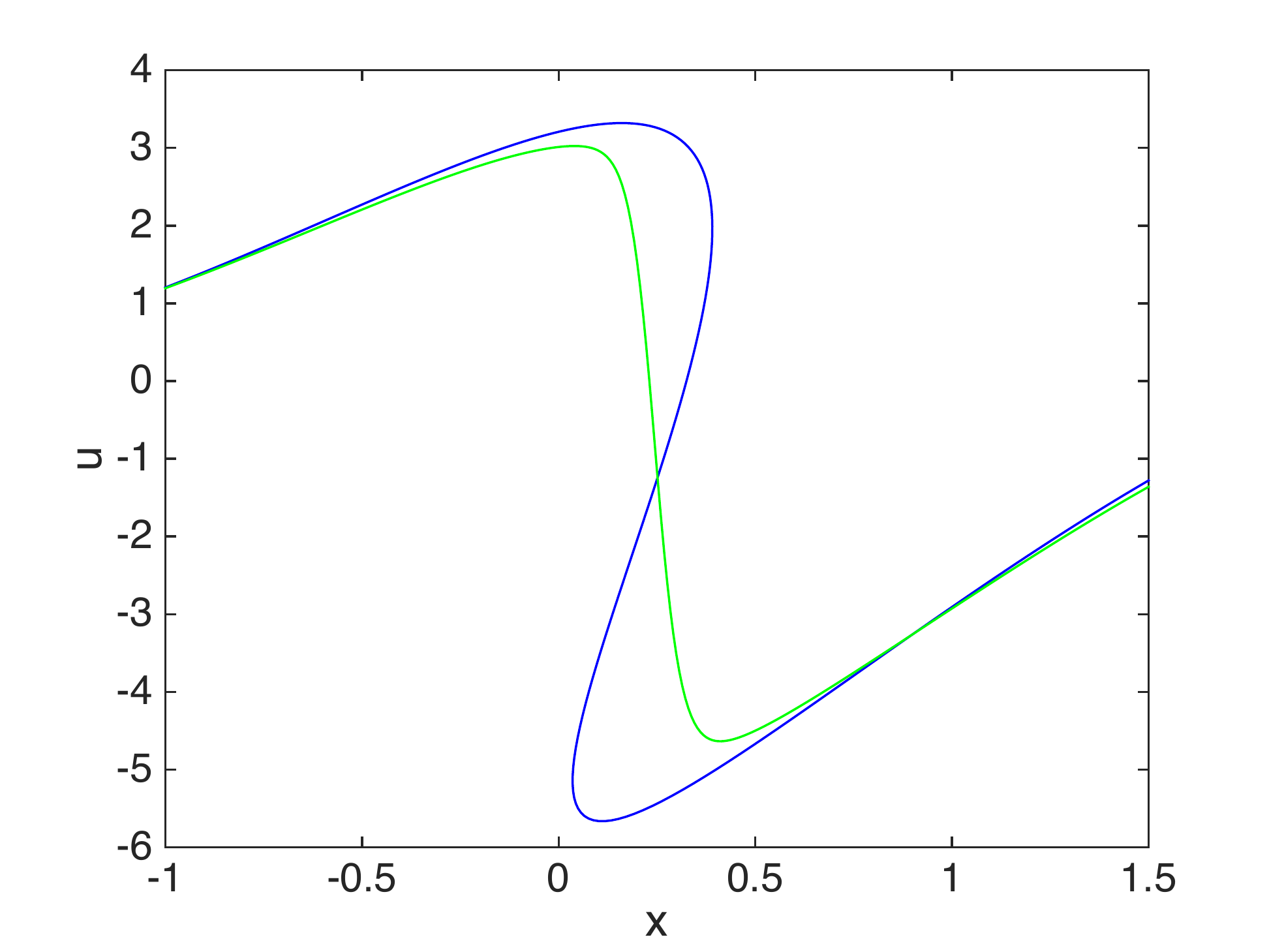}
}
 \caption{Left: numerical solution to the dissipative dKP equation 
(\ref{dKPdis}) with $\epsilon=0.04$, $c=1$, for initial data (\ref{u2}), at 
$t=0.15$. Center: the corresponding Fourier coefficients. Right:
a slice of the left plot along the line $y =-0.4985$ (green), together 
with the corresponding solution of the dKP equation (blue). 
}
 \label{dKPdisexp2}
\end{figure}
For the initial data (\ref{u2}), the code is run with $N_{x}=N_{y}=2^{11}$ 
Fourier modes on the same spatial domain as before, using $N_{t}=2000$ time 
steps for $t\leq 0.15$. The first gradient catastrophe is found at 
$t_{c}=0.087\ldots$, see Table~\ref{table2} for the remaining critical
parameters. The solution at the final time (cf. Fig.~\ref{Fexp2}, left) 
is strongly asymmetric. This also implies an asymmetry of the tails of 
the solution and thus a stronger effect of the algebraic decay of the 
solution towards spatial infinity. The asymmetry of the tails of the solution 
also affects the Fourier coefficients. Despite a higher resolution than 
that of Fig.~\ref{Fexpnorm}, there are small contributions to the high 
wave number Fourier coefficients along the $k_y$ axis above the Krasny filter,
which eventually cause the numerical scheme to break down. As a result,
we do not reach a second catastrophe in this example. At $t=0.15$, 
the relative computed $L^2$ norm is still conserved with an accuracy in 
the order of $10^{-13}$. The $L^{\infty}$ norm of $F$ and of its gradient do 
not indicate blow-up, but they are also not decreasing. If the solution 
exists for large $t$ also, then the computation did not reach the 
asymptotic regime. 

The asymmetry of the solution can also clearly be seen 
from the contour delimiting the multivalued region, seen as the red line 
in Fig.~\ref{Fexp2} (right). This is compared to the asymptotic theory at
$\bar{t} = 0.063$, shown as the blue line. Theory correctly describes the 
strong asymmetry and the orientation of the lip shape, but there are some
quantitative differences. This indicates that the size of the critical region
is smaller in the case of strong asymmetry.

 \begin{figure}
\centering
    \includegraphics[width=0.49\textwidth]{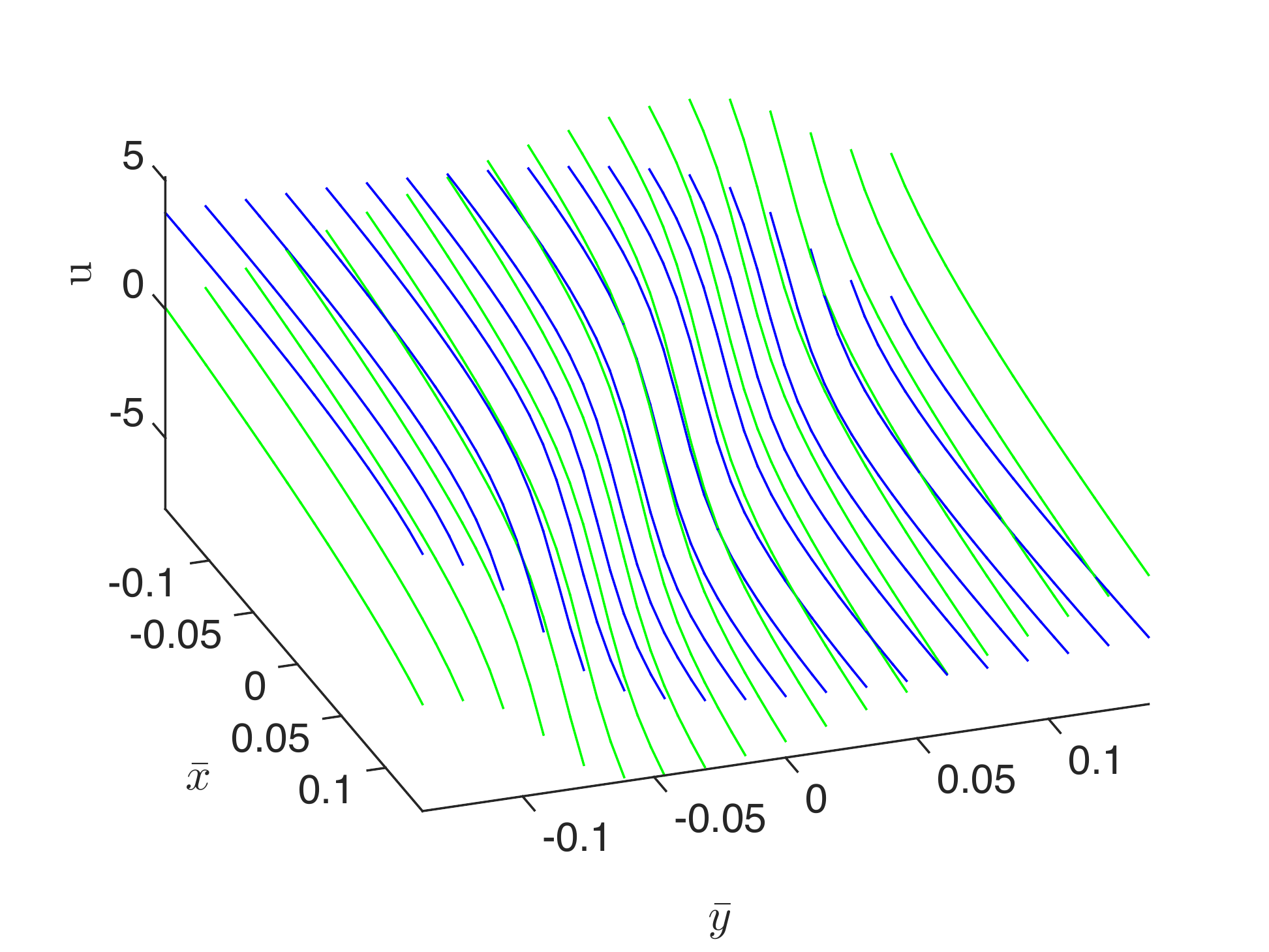} 
  \caption{ In blue the solution of the dissipative dKP equation and 
  in green the Pearcey  asymptotic solution (\ref{KPPearcy}) for 
$\epsilon=0.01$ and the  strongly asymmetric  initial data (\ref{u2}) 
at the critical time $t_{c}$ and near the critical point of the dKP solution. }
 \label{fig:dkpburgerpearceyexpc}
   \end{figure} 
For the dissipative dKP equation for the initial data (\ref{u2}), we 
consider $\epsilon=0.04$ to obtain the solution shown in 
Fig.~\ref{dKPdisexp2} on the left. The Fourier coefficients in the middle of
the same figure are also rather asymmetric, but decrease to the order 
of the Krasny filter. Due to the higher value of $\epsilon$,  the 
loss of the $L^2$ norm is of the order of $22.2\%$. On the right of
Fig.~\ref{dKPdisexp2}, we compare the dissipative solution to the 
corresponding solution of the dKP equation. Although the width of the 
front is greater, owing to a higher value of $\epsilon$, it is set 
inside the s-curve where the shock position is expected to be. 

In Fig.~\ref{fig:dkpburgerpearceyexpc} we show  the dissipative dKP 
equation (\ref{DKP}) for $\epsilon=0.01$ for initial data (\ref{u2}). 
While in the symmetric case $F_{y}^{c} = 0$, here we have 
$F_{y}^{c}\approx-17.39$, consistent with a strongly asymmetric shock. 
Even in this case, the full two-dimensional structure of the step is 
well described by the asymptotic theory. 

\section{Conclusions}
\label{sec:disc}
We have introduced a coordinate transformation, inspired by the method
of characteristics, to investigate wave breaking in the dispersionless 
Kadomtsev-Petviashvili equation. As a result, the entire region where 
the profile is overturned is mapped onto a smooth and single valued function. 
The transformed equation remains smooth near the gradient catastrophe. 
Moreover, our numerics show that solutions remain smooth even beyond 
secondary wave breaking events. This permits 
us to compute solutions up to the first gradient catastrophe with much 
reduced numerical effort, and then to continue into the overturned region, 
where direct numerical simulations of the dKP equation fail. From the 
overturned profile, one can reconstruct the shock position, using the 
jump condition (\ref{shock_cond_final}). 

Using the fact that the transformed profile remains smooth at the gradient
catastrophe, we have calculated the local similarity form of the profile. 
This allows us to calculate the lip shape of the overturned region 
analytically, and to find the position of shock. Both the shape and the 
scaling properties of this region agree well with numerical simulations. 

We have also investigated the dissipative version of the dKP equation, which 
regularizes the gradient catastrophe. We performed direct numerical 
simulations of this equation for small dissipation, which we continued 
beyond the first gradient catastrophe. Results agree with expected shock 
solutions, except that the jump at the shock position is replaced by a 
smooth but rapidly varying profile. To investigate the shape of this 
profile, we use our characteristic transformation to map the dissipative 
KP equation locally to Burgers' equation, which we can solve to obtain 
a local similarity description of the profile in two dimensions.
Asymptotic analysis leads to a description of the profile in terms of 
Pearcey's function, which is in good agreement with numerics. 

We believe that the methods developed in this paper are of interest to
study shock formation in a wider class of hyperbolic equations, including 
the compressible Euler equation. Here a significant complication lies in the 
fact that there are {\it two} families of characteristics in the 
corresponding one-dimensional problem, and hence a transformation 
based on a single characteristic cannot be expected to lead to a solution 
which avoids overturning for all times. However, shocks are generically 
expected to form with respect to one of the two characteristics only 
\cite{LL84a}, so a transformation such as (\ref{implicit0}) will still 
be able to unfold the profile locally. However, the necessary 
transformation will depend on which of the characteristics is involved, 
and thus implicitly on initial conditions. 

\section*{Acknowledgments}
JE's work was supported by a Leverhulme Trust Research Project Grant.
TG  was  partially  supported  by  Miur Research project Geometric and 
analytic  theory  of Hamiltonian  systems  in  finite  and infinite 
dimensions  of  Italian  Ministry  of  Universities and Research.

\providecommand{\bysame}{\leavevmode\hbox to3em{\hrulefill}\thinspace}
\providecommand{\MR}{\relax\ifhmode\unskip\space\fi MR }
\providecommand{\MRhref}[2]{%
  \href{http://www.ams.org/mathscinet-getitem?mr=#1}{#2}
}
\providecommand{\href}[2]{#2}

\end{document}